%% file: ms.tex
\documentclass{article}
\usepackage[utf8]{inputenc}

\usepackage{amsmath,amsfonts,amssymb,color,amsthm}
\usepackage[usenames,dvipsnames,svgnames,table]{xcolor}
\definecolor{darkgreen}{rgb}{0.0,0,0.9}
\usepackage{mathtools}
\usepackage{fullpage}
\usepackage{parskip}
\usepackage{comment}
\usepackage{tikz}
\usepackage{bbm}
\usepackage{dsfont}
\usepackage[sc]{mathpazo}
\usepackage[basic]{complexity}
\usepackage{algorithm}
\usepackage[noend]{algpseudocode}
\usepackage[colorlinks=true,
citecolor=OliveGreen,linkcolor=BrickRed,urlcolor=BrickRed,pdfstartview=FitH]{hyperref}
\usepackage[capitalize,nameinlink]{cleveref}
\usepackage{tcolorbox}
\usepackage[short]{optidef}

\usepackage{multirow}
\usepackage{mathtools}
\usepackage{booktabs}
\usepackage{graphicx}

\let\R\relax
\newcommand*{\R}{\mathbb{R}}

\newcommand*{\Rplus}{\mathbb{R_+}}
\newcommand*{\Qplus}{\mathbb{Q_+}}

\newcommand*{\suppress}[1]{}

\let\poly\relax
\DeclareMathOperator{\poly}{poly}

\makeatletter
\def\thm@space@setup{%
	\thm@preskip= 10pt
	\thm@postskip=\thm@preskip 
}
\makeatother

\makeatletter
\renewcommand{\paragraph}{%
	\@startsection{paragraph}{4}%
	{\z@}{5pt}{-1em}%
	{\normalfont\normalsize\bfseries}%
}
\makeatother

\newtheorem{theorem}{Theorem}
\newtheorem{lemma}[theorem]{Lemma}
\newtheorem{corollary}[theorem]{Corollary}

\newtheorem{claim}[theorem]{Claim}

\theoremstyle{definition}
\newtheorem{definition}[theorem]{Definition}

\newtheorem{remark}[theorem]{Remark}

\newenvironment{fminipage}%
{\begin{Sbox}\begin{minipage}}%
		{\end{minipage}\end{Sbox}\fbox{\TheSbox}}

\newcommand{\CN}{\mbox{${\cal N}$}}
\newcommand{\CS}{\mbox{${\cal S}$}}

\newcommand\vv{\boldsymbol{\mathit{v}}}

\newcommand\RR{\boldsymbol{\mathit{R}}}

\title{Time-Efficient Algorithms for Nash-Bargaining-Based Matching Market Models}
\author{Ioannis Panageas \and Thorben Tr\"obst \and Vijay V.\ Vazirani}
\date{\today}

\begin{document}
    \maketitle

\input{introduction}

\input{premilinaries}

\input{models}

\input{main2.tex}

    \input{conditional_gradient.tex}

    \clearpage
    \bibliographystyle{plain}
    \bibliography{all_refs.bib}
\end{document}

%% file: introduction.tex
\begin{abstract}
    In the area of matching-based market design, existing models using cardinal utilities suffer
    from two deficiencies: First, the Hylland-Zeckhauser (HZ) mechanism \cite{hylland}, which has
    remained a classic in economics for one-sided matching markets, is intractable; computation of
    even an approximate equilibrium is PPAD-complete \cite{VY,HZ-hardness}. Second, there is an
    extreme paucity of such models. This led \cite{HV-NBMM} to define a rich collection of
    Nash-bargaining-based models for one-sided and two-sided matching markets, in both Fisher and
    Arrow-Debreu settings, together with very fast implementations using available solvers and very
    encouraging experimental results.

    In this paper, we give fast algorithms with proven running times for the models of
    \cite{HV-NBMM} using the techniques of multiplicative weights update (MWU) and conditional
    gradient descent (CGD). Additionally, we make the following contributions:

    \begin{enumerate}
        \item By \cite{TV24}, a linear one-sided Nash-bargaining-based matching market
            satisfies envy-freeness within factor two. We show that the other models satisfy
            approximate equal-share fairness, where the exact factor depends on the utility function
            being used in the particular model.
        \item We define a Nash-bargaining-based model for non-bipartite matching markets and give
            fast algorithms for it using conditional gradient descent. 
    \end{enumerate}
\end{abstract}

\section{Introduction}
\label{sec:intro}

For a mechanism to be highly impactful, it must have both good game-theoretic properties and
computational efficiency\footnote{A poster child in this respect is the Gale-Shapley Deferred
Acceptance mechanism \cite{GaleS} for  stable matching; its spectacular success led to the creation
of the area of matching-based market design.}. Our paper deals with cardinal-utility matching
markets, for which the most prominent mechanism, for the key case of a linear one-sided matching
market, was due to Hylland and Zeckhauser \cite{hylland}. This pricing-based mechanism has all the
game-theoretic properties one could ask for; however, it turns out to be computationally intractable
in theory and practice.

Fortunately a viable alternative has emerged, namely a Nash-bargaining-based model \cite{HV-NBMM}.
Very recent work \cite{TV24} has shown that it possesses good game-theoretic properties, and since
it involves solving a convex program, its solution can be found in polynomial time using the
ellipsoid algorithm; however, the latter is far from practical. The purpose of this paper is to
rectify the last deficiency by giving fast implementations with proven running times; for this
purpose we use the techniques of multiplicative weights update (MWU) and conditional gradient
descent (CGD).

General mechanisms for matching markets belong to two classes: whether they use ordinal or cardinal
utility functions. Whereas the former are easier to elicit and work with, the latter are more
expressive, thereby producing higher quality allocations and leading to significant gain in
efficiency, e.g., \cite{Immorlica-IC} give a striking example in which an allocation under cardinal
utilities improves each of the $n$ agents by a factor of $\theta(\log(n))$ over the allocation made
under ordinal utilities (obtained by coarsening the given cardinal utilities). Furthermore, whereas
matching markets for the former  are well-developed from the viewpoint of both theory and practice
\cite{GaleS, Shapley1974-TTC, Bogomolnaia-PS, Moulin2018fair}, the latter had serious deficiencies
which are only now being resolved, including in the current paper.

The Hylland-Zeckhauser (HZ) mechanism \cite{hylland} uses the power of the pricing mechanism to
produce allocations that are Pareto optimal and envy-free \cite{hylland}, and the mechanism is
incentive compatible in the large \cite{He2018pseudo}. However, the recent works \cite{VY,
HZ-hardness} show that computation of even an approximate HZ equilibrium is PPAD-complete. HZ is not
only intractable in theory but also in practice: There is no mathematical object, such as an LCP, a
convex program or even a non-convex program, which captures an HZ equilibrium. The only known method
for computing it uses algebraic cell decomposition \cite{Basu1995, Alaei2017}, which requires
exploration of $n^{5n^2}$ cells, making it prohibitive even for small $n$.  

Besides the intractability of HZ, the area of matching-based market design suffers from a second
deficiency:  There is an extreme paucity of models using cardinal utilities. This stands in sharp
contrast with general equilibrium theory, which defined and extensively studied several fundamental
market models to address a number of specialized and realistic situations, e.g., see
\cite{Mas1995microeconomic}. If we were to draw a parallel with models in general equilibrium
theory, HZ would correspond to the most elementary model, namely the linear Fisher model. Hylland
and Zeckhauser \cite{hylland} attempted an extension of their model to the Arrow-Debreu, also called
exchange, setting; however, they found instances that do not admit an equilibrium. In light of this
difficulty, studying further generalizations, such as two-sided matching market models, made little
sense.

For all the reasons stated above, designing an alternative mechanism for cardinal utility matching markets became a pressing issue. Fortunately, for traditional market models, an alternative to the pricing-based mechanism   had been studied in the past: \cite{va.rational} gave a Nash-bargaining-based mechanism for the linear Arrow-Debreu model; see Section \ref{sec.related} for the point of view that led to this work. Building on this idea, \cite{HV-NBMM} defined a rich collection of Nash-bargaining-based models for one-sided and two-sided matching markets, in both Fisher and Arrow-Debreu settings. Since the Nash bargaining solution is captured by a convex program, it can be efficiently computed to any degree of accuracy using the ellipsoid algorithm \cite{GLS, Vishnoi.book}. Hence \cite{HV-NBMM} addressed both issues, the intractability of HZ and the paucity of matching market models using cardinal utilities. However, since the ellipsoid algorithm is very slow, \cite{HV-NBMM} gave implementations of several of their models using available solvers, and very favorable  experimental results; see Section \ref{sec.related} for details. 

We next address game-theoretic properties of the Nash-bargaining-based mechanism \cite{HV-NBMM} for
the linear one-sided matching market. By definition, the Nash bargaining solution is Pareto optimal;
however easy examples show that it is neither envy-free nor incentive compatible. On the other hand,
recent work of \cite{TV24} has shown that this mechanism satisfies very favorable
properties: they showed that it satisfies envy-freeness within factor two and incentive
compatibility within factor two; moreover, both bounds are tight. Next, using polyhedral methods
they showed that it admits an envy-free and Pareto optimal solution which is moreover rational as
well. That raised the question of finding such a solution efficiently. However, they showed that
this problem is PPAD-hard via a reduction from the problem of computing an approximate HZ solution,
which was shown to be PPAD-hard by \cite{HZ-hardness}; membership in PPAD was shown by
\cite{Caragiannis-complexity}. 

As a result of these new findings, the game-theoretic properties of the one-sided
Nash-bargaining-based model are almost as good as is feasible. Furthermore, as mentioned above,
\cite{HV-NBMM} gave very fast implementations. That still leaves the issue of obtaining fast
algorithms with proven running times; that is accomplished in the current paper.

Another question studied in \cite{TV24} was whether two-sided markets admit envy free
and Pareto optimal allocations; \cite{Bogomolnaia2004random} had given a positive answer for the
case of symmetric dichotomous utilities. However, \cite{TV24} show that on relaxing
either of the conditions, symmetry or dichotomous utilities, such an allocation may not exist. They
gave examples of a two-sided matching market with symmetric $\{0, 1, 2\}$ utilities and another with
asymmetric dichotomous utilities for which such allocations don't exist, hence precluding
pricing-based mechanisms, since the latter yield envy free and Pareto optimal allocations. In
contrast, the Nash-bargaining-based approach \cite{HV-NBMM} easily yields models for these and more
general matching market settings. 

The current paper gives fast algorithms  with proven running times for these more general settings as well. Additionally, we make the following contributions:

\begin{enumerate}
	\item By definition, all the Nash-bargaining-based models satisfy Pareto optimality. However, in the absence of envy-freeness, Pareto optimality does not mean much, since highly skewed allocations become permissible. We rectify this to the extent possible by showing that these models satisfy proportionality,
		where the exact factor depends on the utility function being used in the particular model, see Lemmas \ref{lem:opt_lb},  \ref{lem:opt_lb_splc}, \ref{lem:opt_lb_nb}, \ref{lem:opt_lb_linear_ci} and \ref{lem:opt_lb_all}.
	\item We define, for the first time, a model for non-bipartite matching markets -- it is also Nash-bargaining-based -- and give fast algorithms for it using conditional gradient descent. The roommates problem under cardinal utilities is an application of this model.
\end{enumerate}

Our fast, combinatorial\footnote{We use the term ``combinatorial'' in the same sense as Schrijver \cite{Sch-book}, i.e., the algorithm does not use an LP or convex program solver.} algorithms are based on the techniques of multiplicative weights update (MWU) and conditional gradient descent (CGD); we solve our non-bipartite matching market model using CGD only. In every case, we study not only the Fisher but also the Arrow-Debreu version; the latter is also called the exchange version. Unlike the difficulty of generalizing HZ to the exchange setting, Nash bargaining readily lends itself to this extension, since the disagreement utility encodes, in a natural way, the utility of the initial endowment of each agent.

Since these models are inspired by models in general equilibrium theory, we borrow terminology from that theory to describe them. For each model, we give standard applications as well as the algorithmic technique(s) we use in this paper. As the models get more general, we will only state their additional features.
\begin{enumerate}
	\item {\em One-sided linear Fisher (MWU and CGD):} The setup is analogous to that of HZ, and its standard application is matching agents to goods, with only agents having utilities for goods.
	\item {\em One-sided linear Arrow-Debreu (MWU and CGD):} This is the exchange or Arrow-Debreu version of the previous model. Thus agents start with an initial endowment of goods and exchange them to improve their happiness.
	\item {\em One-sided SPLC Fisher and Arrow-Debreu (MWU and CGD):} Economists model {\em diminishing marginal utilities} via concave utility functions. Since we are in a fixed-precision model of computation, we will consider separable, piecewise-linear concave (SPLC) utility functions, thus generalizing both previous models.
	\item {\em Two-sided linear Fisher and Arrow-Debreu (CGD):} The standard application is matching workers with firms, with each side having utility functions over the other.
	\item {\em Two-sided SPLC Fisher and Arrow-Debreu (CGD):} These are the SPLC generalizations of the previous models.
	\item {\em Non-bipartite linear Fisher and Arrow-Debreu (CGD):} The standard application is the roommates problem, i.e., pairing up students for rooms in a dormitory. In the Arrow-Debreu version, the agents currently have rooms in the dormitory -- their initial endowment -- and wish to move to better rooms for the next academic year.
\end{enumerate}

The recent computer science revolutions of the Internet and mobile computing led to the launching of highly impactful and innovative one-sided and two-sided matching markets, such as online advertisement platforms (Google Ads), ride-hailing services (Uber, Lyft), food delivery services (Doordash, Uber Eats), vacation rentals (Airbnb, VRBO), freelancing (Taskrabbit, Upwork) and online dating services (Match.com, OkCupid). In turn it led to the formation of the new area of online and matching-based market design, e.g., see \cite{Simons} and \cite{MM-book}. Our paper is motivated by this challenge and opportunity.

\subsection{Technical Contributions}
\label{sec.contributions}

In this section, we highlight the novel technical ideas of this paper. These fall into three categories.

{\bf 1). Modeling:}
Because of the difficulties encountered in generalizing HZ, mentioned in the Introduction, the economics literature has not defined a model of matching markets for general graphs, in particular, non-bipartite graphs, under cardinal utilities, despite the fact that it has natural applications. Our Nash-bargaining-based model for non-bipartite matching markets, given in Section \ref{sec.models}, is the first of its kind.

{\bf 2). Multiplicative Weights Update (MWU):}

 We provide a MWU algorithm which gives an $\epsilon$-approximate Nash bargaining solution\footnote{For a formal definition, see Section \ref{sec:mwua}.} in time $\frac{n^3 \log n}{\epsilon^2}$. The algorithm follows the primal-dual scheme and therefore yields both allocations and prices. Our techniques deviate from the literature due to additional constraints, namely matching, and the presence of endowments.

 The main idea underlying our approach is to create a feasibility program, see (\ref{program:1}), the feasible solutions of which are exactly the optimal primal and dual values of the variables of Program (\ref{eq.CP-LAD}). We then find an approximate feasible solution to this program  using MWU, yielding an approximate Nash bargaining solution. Although other works in the past have also followed this approach of first finding a feasibility program, e.g., to solve linear programs and find equilibria in markets \cite{Arora12}, \cite{Fleischer2008fast}, in our case this step was much more challenging as expounded below.

In (\ref{program:1}), each agent aims to maximize her utility subject to budget constraints, i.e.,  for each agent $i$ the following holds: $\sum_{j} x_{ij}(p_j+q_i) \leq 1+c_{i} \min _{j } \frac{u_{ij}}{p_j+q_i}.$ Observe that both the left and the right hand sides of the inequality involve prices; in particular, in the case of Linear Fisher Program (\ref{eq.EG-LF}), the right hand side is just 1. To come up with the right budget constraints, we used the KKT conditions.  The KKT conditions would yield that each agent $i$ maximizes her utility $u_i$ subject to budget constraints that would involve the utility $u_i$ itself. This self-reference had to be removed.

 Note that the aforementioned challenge becomes even harder for the case of SPLC utilities where the correct constraints would involve prices in addition to allocations for all segments in the utility functions. The rest of the analysis of MWU algorithm follows a potential function argument adapted appropriately to our setting.

{\bf 3). Conditional Gradient Descent (CGD):}

First order methods, in particular projected gradient descent, have proven successful for solving large scale market equilibrium problems. The challenges to extending this approach to our models are two-fold. We get around these hurdles by resorting to a conditional gradient method instead.

{\bf a).}  The objective function of each of the convex programs we study is logarithmic. Therefore, its gradient is neither Lipschitz nor smooth. It must therefore be modified in order to guarantee these properties. \cite{gao2020first} circumvent this problem for the Fisher market by replacing the objective by its quadratic extension below a certain point. This is obtained through the well-known equal-share fairness property, i.e.\ the fact that in an equilibrium we have $u_i \geq \frac{1}{n} \sum_{j \in G} u_{i j}$ for all agents $i$. Since our models are more involved, we need to extend this property by introducing approximations. We show that our models satisfy approximate notions of equal-share fairness, which allows us to bound the equilibrium solution away from the boundary. In particular, this holds even in the case of Arrow-Debreu extension of models, i.e., with initial endowments. We also show that these bounds are sharp up to a constant factor, which may be of independent interest.

{\bf b).} In the Fisher setting, the efficiency of the projected gradient descent method relies on  projections onto the simplex which are not time-consuming. On the other hand, in our models, the feasible region is generally given by matching or flow polytopes, for which projections are significantly more expensive. Therefore, we employ a conditional gradient method instead which relies on combinatorial matching and flow algorithms that are very efficient in practice. In the case of SPLC utilities we show how an additional ``shifting step'' can be used to decrease the dependence on the (rather large) diameter of the feasible region.

\subsection{Related Results}
\label{sec.related}

As stated in the Introduction, the first paper to suggest the use of a Nash-bargaining-based mechanism for solving a market model which had traditionally been addressed via the pricing mechanism was \cite{va.rational} ---  for the linear Arrow-Debreu model. We explain below the logical steps that led to it, since this idea is the source of the Nash-bargaining-based matching market models. 

The starting point was a remarkable convex program due to Eisenberg and Gale \cite{eisenberg}, given in (\ref{eq.EG-LF}). Its optimal solution gives equilibrium allocations for the linear Fisher model (see definition in Section \ref{sec.preliminaries}) and optimal dual solution gives equilibrium prices. This program can be solved to $\epsilon$ precision in time that is polynomial in the size of the input and $\log{1/\epsilon}$ via ellipsoid-based methods \cite{GaleS, Vishnoi.book}. Additionally, since the Eisenberg-Gale program is a rational convex program (see definition below), it can be solved exactly in polynomial time using ideas from \cite{Jain-AD}.

	\begin{maxi!}
        {} {\sum_{i \in A}  {\log \left(\sum_{j \in G}  {u_{ij} x_{ij}}\right)}}
			{\label{eq.EG-LF}}
		{}
		\addConstraint{\sum_{i \in A} {x_{ij}}}{\leq 1 \label{const.2EG}}{\quad \forall j \in G}
		\addConstraint{x}{\geq 0.}
	\end{maxi!}

This program has the following features:

\begin{enumerate}
	\item In the absence of this program, we would need to solve $n$ linear programs, one for each agent, to compute equilibrium allocations, {\em after we are given equilibrium prices}. On the other hand, using KKT conditions one can show that the Eisenberg-Gale (EG) program single-handedly computes not only the equilibrium prices but also all $n$ equilibrium allocations. 
	\item The KKT conditions also reveal that this program always has a rational optimal solution if all parameters are rational, i.e., it is ``behaving'' like a linear program. Similar to an LP, the number of bits needed to write its solution is also polynomial in the size of the input. 
	\item If we use Nash bargaining for solving the linear Fisher model, the convex program we would obtain is precisely the EG program. 
\end{enumerate}

Given such strong properties, it was natural to look for similar convex programs for more general market models. A slight generalization of the linear Fisher model is the linear Arrow-Debreu (exchange) model; however there does not seem to be a generalization of the EG program which captures equilibrium allocation for the latter model. On the other hand, one can attempt to solve this market model via Nash bargaining; this will yield a convex program which again has the property of admitting a rational optimal solution, similar to the EG program. The latter property suggests solving this convex program via a combinatorial polynomial time algorithm --- this was accomplished in \cite{va.rational}. Furthermore, it gave the idea of defining the notion of a {\em rational convex program}, i.e., one that admits a rational optimal solution if all parameters are rational, see \cite{va.rational}.

The first comprehensive study of the computational complexity of the HZ scheme was undertaken in \cite{VY}. They gave an example for which the HZ equilibrium is irrational, thereby showing that the problem of computing an exact HZ equilibrium is not in PPAD. Most importantly, they showed that the problem of approximately computing an equilibrium is in PPAD, though they left open the problem of showing PPAD-hardness. The latter was established in \cite{HZ-hardness}.

This intractability opened up the question of finding tractable models of matching markets under cardinal utilities. \cite{HV-NBMM} answered this question as detailed in the Introduction. In particular, they gave very general models of matching markets for which the resulting convex program has linear constraints thereby ensuring zero duality gap and polynomial time solvability. However, as is well known, polynomial time solvability is often just the beginning of the process of obtaining an ``industrial grade''  implementation. Towards this end, \cite{HV-NBMM} gave very fast implementations, using available solvers, and left the open problem of giving implementations with proven running times. The implementation given in \cite{HV-NBMM} can solve very large instances, with $n = 2000$, in one hour even for a two-sided matching market.

To deal with the fact that the Arrow-Debreu extension of HZ does not always admit an equilibrium, \cite{Echenique2019constrained} defined the notion of an $\alpha$-slack Walrasian equilibrium. This is a hybrid between the Fisher and Arrow-Debreu settings. Agents have initial endowments of goods and for a fixed $\alpha \in (0, 1]$, the budget of each agent, for given prices of goods, is $\alpha + (1 - \alpha) \cdot m$, where $m$ is the value for her initial endowment; the agent spends this budget to obtain an optimal bundle of goods. Via a non-trivial proof, they showed that for $\alpha > 0$, an $\alpha$-slack Walrasian equilibrium always exists. \cite{Garg-ADHZ} gave the notion of an $\epsilon$-approximate Arrow-Debreu HZ equilibrium and using the result stated above on $\alpha$-slack Walrasian equilibrium, they showed that such an equilibrium always exists. They also gave a polynomial time algorithm for computing such an equilibrium for the dichotomous and bivalued utility functions.

\cite{Abebe-MM-Truthful} gave an incentive-compatible mechanism for cardinal-utility matching markets and they showed that it is possible to give every agent a $O(2^{2 \sqrt{\log n}})$ fraction of the utility which they would have gotten from Nash bargaining in a truthful way; however, their mechanism itself does not use Nash bargaining.

In recent years, several researchers have proposed Hylland-Zeckhauser-type mechanisms for a number of applications, e.g., see \cite{Budish2011combinatorial, He2018pseudo, Le2017competitive, Mclennan2018efficient}. The basic scheme has also been generalized in several different directions, including two-sided matching markets, adding quantitative constraints, and to the setting in which agents have initial endowments of goods instead of money, see  \cite{Echenique2019constrained, Echenique2019fairness}.

A large number of algorithms have been developed for computing a Fisher market equilibrium. Notable examples are the DPSV algorithm \cite{DPSV}, which is combinatorial, and an algorithm based on interior point method \cite{ye2008path}, which converges in time $O(\mathrm{poly}(n) \log (1 / \epsilon))$.

\paragraph{ \ \ Related works using Multiplicative Weights Update (MWU).} The Multiplicative Weights Update (MWU) is a ubiquitous meta-algorithm with numerous applications in different fields \cite{Arora12}. More specifically, it has been used in max-flow problems \cite{mwuflows}, discrepancy minimization \cite{ipcorothvoss}, learning graphical models \cite{mwugraphicalmodels}, even in evolution \cite{CLPV14, MPP15}. It is particularly useful in algorithmic game theory due to its regret-minimizing properties \cite{Fudenberg98,Cesa06}, i.e., the time average behavior of MWU leads to (approximate) coarse correlated equilibria (CCE).

MWU has also been used to compute market equilibria in the Linear Fisher model (\ref{eq.EG-LF}); most notably is the work in \cite{Fleischer2008fast}. They give a simple and decentralized algorithm based on the multiplicative weights update method, though its running time is $O(\mathrm{poly}(n)  / \epsilon^2)$.
Due to the general nature of multiplicative weights, their algorithm extends to much more general classes of utilities including some that do not satisfy weak gross substitutability.
We will use similar techniques, since MWU can be adapted to handle the matching constraints in our models.
However, a non-trivial extension of the algorithm is required in order to deal with initial endowments.

\paragraph{\ \ Related works using Gradient Descent (GD).} Arguably one of the most commonly used first-order methods for minimizing differentiable objectives, which has become very popular in optimization, machine learning and computer science communities is Gradient Descent. The main reason behind this fact lies in GD's simplicity and nice properties. For convex objectives, one can show that as long as the function is Lipschitz, $O(1/\epsilon^2)$ steps suffice to get an $\epsilon$-approximate solution. Moreover, if the function has Lipschitz gradient\footnote{Which we enforce in our case.} then $O(1/\epsilon)$ steps suffice and finally if the function is strongly-convex, one can get an $\epsilon$-approximate optimum in $O(\log (1/\epsilon))$ iterations, see \cite{bubeck15} for more information.

Most recently, GD and its stochastic counterpart have been extensively studied and used for optimizing non-convex landscapes with the guarantee of convergence almost always to local optima \cite{Ge15, Lee19, Jin21}. These results shed light on why GD works well in practice.

GD has also found numerous applications for computing market equilibria, most notably the recent work of \cite{gao2020first}, see references therein. \cite{gao2020first} studied first-order methods for the Fisher market by considering gradient ascent type algorithms for the various convex programming formulations.
They show that one may exploit the fairness of the Fisher market in order to bound the market equilibrium away from the boundary of the feasible region.
Using this they develop a projected gradient ascent algorithm that also converges in $O(\poly(n) \log(1 / \epsilon))$ time and requires only simplex projections.
However, the efficiency of this approach depends quite critically on the simple structure of the convex programs for Fisher markets. A more general algorithm of this form would depend on geometric properties of the feasible region and use projections which are expensive. For these reasons we resort to a conditional gradient ascent algorithm instead.

%% file: premilinaries.tex
\section{Basic Solution Concepts}
\label{sec.preliminaries}

{\bf Nash Bargaining Game:}
An {\em $n$-person Nash bargaining game} consists of a pair $(\CN, c)$, where $\CN \subseteq \R_+^n$ is a compact, convex set and $c \in \CN$. The set $\CN$ is called the {\em feasible set} -- its elements are vectors whose components are utilities that the $n$ players can simultaneously accrue. Point $c$ is the {\em disagreement point} -- its components are utilities which the $n$ players accrue if they decide not to participate in the proposed solution. 

The set of $n$ agents will be denoted by $A$ and the agents will be numbered $1, 2, \ldots n$. Instance $(\CN, c)$ is said to be {\em feasible} if there is a point in $\CN$ at which each agent does strictly better than her disagreement utility, i.e., $\exists \vv \in \CN$ such that $\forall i \in A, \ v_i > c_i$, and {\em infeasible} otherwise. In game theory it is customary to assume that the given Nash bargaining problem $(\CN, c)$ is feasible; we will make this assumption as well.

The solution to a feasible instance is the point $\vv \in \CN$ that satisfies the following four axioms:

\begin{enumerate}
\item
{\em  Pareto optimality:}  No point in $\CN$ weakly dominates $\vv$.
\item
{\em  Symmetry:} If the players are renumbered, then a corresponding renumber the coordinates of $\vv$ is a solution to the new instance.
\item
{\em  Invariance under affine transformations of utilities:} If the utilities of any player are redefined by
multiplying by a scalar and adding a constant, then the solution to the transformed problem is obtained by
applying these operations to the particular coordinate of $\vv$.
\item
{\em  Independence of irrelevant alternatives:} If $\vv$ is the solution to $(\CN, c)$, and 
$\CS \subseteq \R_+^n$ is a compact, convex set satisfying $c \in \CS$ and  $\vv \in \CS \subseteq \CN$, then $\vv$ is also the solution to $(\CS, c)$.
\end{enumerate}

Via an elegant proof, Nash proved:

\begin{theorem}
[Nash \cite{Nash-Bargaining}]
\label{thm.nash}
If the game $(\CN, c)$ is feasible then there is a unique point in $\CN$ satisfying the axioms stated above. Moreover, this point is obtained by maximizing $\Pi_{i \in A}  {(v_i - c_i)}$ over $\vv \in \CN$.
\end{theorem}

Nash's solution to his bargaining game involves maximizing a concave function over a convex domain, and is therefore the optimal solution to the following convex program.

	\begin{maxi}
		{} {\sum_{i \in A}  {\log (v_i - c_i)}}
			{\label{eq.CP-Nash}}
		{}
		\addConstraint{}{\vv \in \CN}
	\end{maxi}

Further, Nash showed that the unique solution to his game is obtained by maximizing the product of utilities of agents over the feasible set. As a consequence, the allocations it produced by this solution are remarkably fair, see \cite{Nash-Unreasonable, Abebe-MM-Truthful, Moulin2018fair} for remarks to this effect. This issue has been further explored under the name of Nash Social Welfare \cite{cole2018approximating, cole2017convex}. Additionally, since the Nash bargaining solution is captured via a convex program, if for a specific game, a separation oracle can be implemented in polynomial time, then using the ellipsoid algorithm one can get as good an approximation as desired in time polynomial in the number of bits of accuracy needed \cite{GLS, Vishnoi.book}.

{\bf Fisher Market Model:}
The {\em Fisher market model} consists of a set $A = \{1, 2, \ldots n\}$ of agents and a set $G = \{1, 2, \ldots, m\}$ of infinitely divisible goods. By fixing the units for each good, we may assume without loss of generality that there is a unit of each good in the market. Each agent $i$ has money $m_i \in \Qplus$.  

Let $x_{ij}, \ 1 \leq j \leq m$ represent a {\em bundle of goods allocated to agent $i$}. Each agent $i$ has a utility function $u: \R_+^m \rightarrow \Rplus$ giving the utility accrued by $i$ from a bundle of goods. We will assume that $u$ is concave and weakly monotonic. Each good $j$ is assigned a non-negative price, $p_j$. Allocations and prices, $x$ and $p$, are said to form an {\em equilibrium} if each agent obtains a utility maximizing bundle of goods at prices $p$ and the {\em market clears}, i.e., each good is fully sold to the extent of one unit and all money of agents is fully spent. We will assume that each agent derives positive utility from some good and for each agent, there is a good which gives her positive utility; clearly, otherwise we may remove that agent or good from consideration. 

\bigskip

{\bf Arrow-Debreu Market Model:}
The Arrow-Debreu market model, also known as the {\em exchange model} differs from Fisher's model in that agents come to the market with initial endowments of good instead of money. The union of all goods in initial endowments are all the goods in the market. Once again, by redefining the units of each good, we may assume that there is a total of one unit of each good in the market. The utility functions of agents are as before. The problem now is to find non-negative prices for all goods so that if each agent sells her initial endowment and buys an optimal bundle of goods, the market clears. Clearly, if $p$ is equilibrium prices then so is any scaling of $p$ by a positive factor.

\bigskip

{\bf One-Sided Matching Market:}
Let $A = \{1, 2, \ldots n\}$ be a set of $n$ agents and $G = \{1, 2, \ldots, n\}$ be a set of $n$ indivisible goods. Each agent $i$ has a utility function $u_{ij}$ over goods $j \in G$. The goal is to design a mechanism to allocate exactly one good to each agent, so that properties such as Pareto optimality are satisfied. Goods are rendered  divisible by assuming that there is one unit of probability share of each good. Let $x_{ij}$ be the allocation of probability share that agent $i$ receives of good $j$. Then, $\sum_j {u_{ij} x_{ij}}$ is the {\em expected utility} accrued by agent $i$. 

To make it a matching market, an additional constraint is that the total probability share allocated to each agent is one unit, i.e., the entire allocation must form a {\em fractional perfect matching} in the complete bipartite graph over vertex sets $A$ and $G$. Now, a solution can be viewed as a doubly stochastic matrix. The Birkhoff-von Neumann procedure \cite{Birkhoff1946tres, von1953certain} then extracts a random underlying perfect matching in such a way that the expected utility accrued to each agent from the integral perfect matching is the same as from the fractional perfect matching. Since {\em ex ante} Pareto optimality implies {\em ex post} Pareto optimality, the integral allocation will also be Pareto optimal.

{\bf Two-Sided Matching Market:} 
Our two-sided matching market model consists of a set $A = \{1, 2, \ldots n\}$ of agents and a set $J = \{1, 2, \ldots, n\}$ of jobs. Each agent $i$ has a utility function $u_{ij}$ over goods $j \in G$ and each job $j$ has a utility function $w_{ij}$ over agents $i \in A$.  Let $x_{ij}$ be the allocation of probability share that agent $i$ receives of good $j$. Then, $\sum_j {u_{ij} x_{ij}}$ is the expected utility accrued by agent $i$ and $\sum_i {w_{ij} x_{ij}}$ is the expected utility accrued by job $j$. As in the one-sided case, the Birkhoff-von Neumann procedure \cite{Birkhoff1946tres, von1953certain} then extracts a random underlying perfect matching from $x$.
 
{\bf Non-Bipartite Matching Market:}
Our non-bipartite matching market model consists of a set $A = \{1, 2, \ldots, n\}$ of agents and utilities for each pair $u_{ij} \in \RR_+$ which $i$ and $j$ accrue on getting matched to each other. Let $x_{ij}$ be the fractional extent to which $i$ and $j$ are matched. Then, $\sum_j {u_{ij} x_{ij}}$ is the utility accrued by agent $i$. The allocation vector $x \in \R_+^{E}$ is said to be a {\em fractional matching in $G$} if it is a convex combination of (integral) matchings in $G$. As in the bipartite case, it is possible to take a fractional matching and, in  combinatorial, polynomial time, decompose it into a convex combination of most $n^2$ integral matchings. This was shown by Padberg and Wolsey \cite{padberg1984fractional}; for a modern proof see \cite{vazirani2020extension}. This allows a similar rounding strategy if one desires integral allocations.

\subsection{Hylland-Zeckhauser Mechanism}
\label{sec.HZ}

A {\em one-sided matching market} consists of two types of entities, say agents and goods, with only one side having preferences over the other, i.e., agents over goods. Let $A = \{1, 2, \ldots n\}$ be a set of $n$ agents and $G = \{1, 2, \ldots, n\}$ be a set of $n$ indivisible goods. The goal of the HZ mechanism is to allocate exactly one good to each agent. However, in order to use the power of a pricing mechanism, which endows the HZ mechanism with the properties of Pareto optimality and incentive compatibility in the large, it casts this one-sided matching market in the mold of a linear Fisher market as follows. 

Goods are rendered  divisible by assuming that there is one unit of probability share of each good, and utilities $u_{ij}$s are defined as in a linear Fisher market. Let $x_{ij}$ be the allocation of probability share that agent $i$ receives of good $j$. Then, $\sum_j {u_{ij} x_{ij}}$ is the {\em expected utility} accrued by agent $i$. Each agent has 1 dollar for buying these probability shares and each good $j$ has a price $p_j \geq 0$.

Beyond a Fisher market, an additional constraint is that the total probability share allocated to each agent is one unit, i.e., the entire allocation must form a {\em fractional perfect matching} in the complete bipartite graph over vertex sets $A$ and $G$. Subject to these constraints, each agent buys a utility maximizing bundle of goods. Another point of departure from a linear Fisher market is that in general, an agent's optimal bundle may cost less than one dollar, i.e., the agents are not required to spend all their money. Since each good is fully sold, the market clears. Hence these are defined to be {\em equilibrium allocation and prices}.

Clearly, an equilibrium allocation can be viewed as a doubly stochastic matrix. The Birkhoff-von Neumann procedure then extracts a random underlying perfect matching in such a way that the expected utility accrued to each agent from the integral perfect matching is the same as from the fractional perfect matching. Since {\em ex ante} Pareto optimality implies {\em ex post} Pareto optimality, the integral allocation will also be Pareto optimal. 

Next we describe the structure of optimal bundles from \cite{VY}. Let $p$ be given prices which are not necessarily equilibrium prices. By the definition of an optimal bundle given above, the optimal bundle for agent $i$ is a solution to LP (\ref{eq.HZ-primal}).

	\begin{maxi}
		{} {\sum_{j \in G}  {u_{ij} x_{ij}}}
			{\label{eq.HZ-primal}}
		{}
		\addConstraint{\sum_{j} {x_{ij}}}{=1}{}
		\addConstraint{\sum_{j} {p_j x_{ij}}}{\leq 1 }{}
		\addConstraint{x_{ij}}{\geq 0}{\ \ \forall j \in G}
	\end{maxi}

%% file: models.tex
\section{Nash-Bargaining-Based Models}
\label{sec.models}

We will define four one-sided matching market models based on our Nash bargaining approach. For the case of linear utilities, we have defined both Fisher and Arrow-Debreu versions, namely {\em LiF} and {\em LiAD}. For more general utility functions we have defined only the Arrow-Debreu version; the Fisher version is obtained by setting all disagreement utilities to zero. For each model, we give the convex program whose optimal solution captures the Nash bargaining solution.

Our one-sided matching market models consist of a set $A = \{1, 2, \ldots n\}$ of agents and a set $G = \{1, 2, \ldots, n\}$ of infinitely divisible goods; observe that there is an equal number of agents and goods. There is one unit of each good and each agent needs to be allocated a total of one unit of goods. Hence the allocation needs to be a fractional perfect matching, as defined next.

\begin{definition}
	\label{def.pm}
	Let us name the coordinates of a vector $x \in \R_+^{n^2}$ by pairs $i, j$ for $i \in A$ and $j \in G$. Then $x$ is said to be a {\em fractional perfect matching} if
	\[ \forall i \in A: \ \sum_j {x_{ij}} = 1 \ \ \ \mbox{and} \ \ \ \forall j \in G: \ \sum_i {x_{ij}} = 1 . \]
\end{definition}

As mentioned in Section \ref{sec.preliminaries}, an equilibrium allocation can be viewed as a doubly stochastic matrix, and the Birkhoff-von Neumann procedure \cite{Birkhoff1946tres, von1953certain} can be used to extract a random underlying perfect matching in a way that the expected utility accrued to each agent from the integral perfect matching is the same as from the fractional perfect matching.

{\bf 1).} Under the {\em linear Fisher Nash bargaining one-sided matching market}, abbreviated {\em 1LF}, each agent $i \in A$ has a linear utility function, $u_i(x) = \sum_{j \in G} u_{i j} x_{i j}$. Corresponding to each fractional perfect matching $x$, there is a vector $v_x$ in the feasible set $\CN$; its components are the utilities derived by the agents under the allocation given by $x$. The disagreement point is the origin. Observe that the setup of {\em 1LF} is identical to that of the HZ mechanism; the difference lies in the definition of the solution to an instance. (\ref{eq.CP-LF}) is a convex program for {\em LiF}.

	\begin{maxi!}
        {} {\sum_{i \in A}  {\log \left(\sum_{j \in G}  {u_{ij} x_{ij}}\right)}}
			{\label{eq.CP-LF}}
		{}
		\addConstraint{\sum_{i \in A} {x_{ij}}}{= 1 \label{const.1}}{\quad \forall j \in G}
		\addConstraint{\sum_{j \in G} {x_{ij}}}{= 1 \label{const.2}}{\quad \forall i \in A}
		\addConstraint{x}{\geq 0.}
	\end{maxi!}

{\bf 2).} Under the {\em linear Arrow-Debreu Nash bargaining one-sided matching market}, abbreviated {\em LiAD}, each agent $i \in A$ has a linear utility function, as above. Additionally, we are specified an initial fractional perfect matching $x_I$ which gives the initial endowments of the agents. Each agent has one unit of initial endowment over all the goods and the total endowment of each good over all the agents is one unit, as given by $x_I$. These two pieces of information define the utility accrued by each agent from her initial endowment; this is her disagreement point $c_i$. We will assume that the problem is feasible, i.e., there is a fractional perfect matching, defining a redistribution of the goods, under which each agent $i$ derives strictly more utility than $c_i$. Each vector $v \in \CN$ is as defined in {\em LiF}. Henceforth, we will consider the slightly more general problem in which are specified the disagreement point $c$ and not the initial endowments $x_I$. There is no  guarantee that $c$ comes from a valid fractional perfect matching of initial endowments. However, we still want the problem to be feasible. 

\begin{remark}
	Note that in all the algorithms presented in this paper, we will provide explicit efficient algorithms for testing for feasibility. Secondly, we will relax the equalities in the convex programs to inequalities in order to ensure that the corresponding dual variables are constrained to be non-negative. It is easy to see that the left-over goods can be ``packed'' into the left-over demand of agents, without changing the objective function, to yield a fractional perfect matching, 
\end{remark}

The convex program for {\em LiAD} is given in (\ref{eq.CP-LAD}).
	\begin{maxi!}
        {} {\sum_{i \in A}  {\log (u_i(x) - c_i)}}
			{\label{eq.CP-LAD}}
		{}
		\addConstraint{\sum_{i \in A} {x_{ij}}}{\leq1 \label{const-LAD.2}}{\quad \forall j \in G}
		\addConstraint{\sum_{j \in G} {x_{ij}}}{\leq1 \label{const-LAD.1}}{\quad \forall i \in A}
		\addConstraint{x}{\geq 0.}
	\end{maxi!}
	
The KKT conditions for program (\ref{eq.CP-LAD}) with non-negative dual variables $p_j$ and $q_i$ are given below. By setting $c_i = 0$ for $i \in A$, we get KKT conditions for program (\ref{eq.CP-LF}) as well.

\begin{description}
\item [(1)]
$ \forall j \in G: \ p_j > 0  \  \implies \ \sum_{i \in A} {x_{ij}} = 1$.
\item [(2)]
$\forall i \in A: \ q_i > 0  \  \implies \ \sum_{j \in G} {x_{ij}} = 1$.
\item [(3)]
    $ \forall i \in A, \ \forall j \in G: \ p_j + q_i \geq \frac{u_{ij}}{u_i(x) - c_i}$.

\item [(4)]
    $ \forall i \in A, \ \forall j \in G:  x_{ij} > 0 \ \implies p_j + q_i = \frac{u_{ij}}{u_i(x) - c_i}$.
\end{description}

For $c_i = 0$, this allows us to prove the following approximate equal-share fairness property which guarantees that every agent achieves at least $1 / 2$ of the utility that they would get under the equal share matching that assigns $\frac{1}{n}$ to all edges.

\begin{lemma}\label{lem:opt_lb}
    Let $x$ be an optimal solution to \eqref{eq.CP-LF} (or rather $\eqref{eq.CP-LAD}$ with $c_i = 0$), then for all agents $i$ we have
    \[
        u_i(x) = \sum_{j \in G} u_{i j} x_{i j} \geq \frac{1}{2 n} \sum_{j \in G} u_{i j}
    \]
\end{lemma}

\begin{proof}
    By the KKT conditions we know that $\frac{u_{i j}}{u_i(x)} \leq p_j + q_i$ with equality if $x_{i j} > 0$.
    Moreover, if $p_j > 0$ then $\sum_{i \in A} x_{i j} = 1$ and likewise for $q_i$.
    Thus
    \[
        \sum_{j \in G} p_j + \sum_{i \in A} q_i = \sum_{i \in A} \sum_{j \in G} {x_{i j} (p_j + q_i)} \leq n.
    \]
    In addition, note that $q_i \leq 1$ since otherwise we cannot have $\sum_{j \in G} x_{i j} = 1$ and $\sum_{j \in G}{x_{i j} (p_j + q_i)} = 1$.
    Finally, we conclude
    \begin{align*}
        u_i(x) &= \max \left\{\frac{u_{i j}}{p_j + q_i} \;\middle|\; j \in G\right\} \\
            &\geq \sum_{j \in G} \frac{u_{i j}}{p_j + q_i} \cdot \frac{p_j + q_i}{\sum_{j' \in G}{p_{j'} + q_i}} \\
            &\geq \frac{\sum_{j \in G} u_{i j}}{n q_i + \sum_{j \in G} p_j} \\
            &\geq \frac{1}{2 n} \sum_{j \in G} u_{i j}. \qedhere
    \end{align*}
\end{proof}

In the case of non-zero $c_i$, the equal share matching may no longer be feasible and so this notion loses some meaning.
However, we will give an analogous weaker bound in Section~\ref{sec.cond_grad}, specifically in Lemma~\ref{lem:opt_lb_linear_ci}.

{\bf 3).} Economists like to model {\em diminishing marginal utilities for goods} by considering concave utility functions. Since we are in a fixed-precision model of computing, we will consider separable, piecewise-linear concave (SPLC) utility functions.

The {\em separable, piecewise-linear concave Arrow-Debreu Nash bargaining one-sided matching market}, abbreviated {\em SAD}, is analogous to LiAD, with the difference that each agent has a separable, piecewise-linear concave utility function, hence generalizing the linear utility functions specified in LiAD. When there are no disagreement utilities $c_i$, we call this model \emph{SF}. We next define these functions in detail.

For each agent $i$ and good $j$, function $f_i^j: \Rplus \rightarrow \Rplus$ gives the utility derived by $i$ as a function of the amount of good $j$ she receives. Each $f_i^j$ is a non-negative, non-decreasing, piecewise-linear, concave function. The overall utility of buyer $i$, $u_i(x)$, for bundle $x=(x_1,\ldots,x_n)$ of goods, is additively separable over the goods, i.e., $u_i(x) = \sum_{j \in G} f_i^j(x_j)$.

We will call each piece of $f_i^j$ a {\em segment}. Number the segments of $f_i^j$ in order of decreasing slope; throughout we will assume that these segments are indexed by $k$ and that $S_{ij}$ is the set of all such indices. Let $\sigma_{ijk}, \ k \in S_{ij}$, denote the $k^{th}$ segment, $l_{ijk}$ denote the amount of good $j$ represented by this segment; we will assume that the last segment in each function is of unbounded length. Let $u_{ijk}$ denote the rate at which $i$ accrues utility per unit of good $j$ received, when she is getting an allocation corresponding to this segment. Clearly, the maximum utility she can receive corresponding to this segment is $u_{ijk} \cdot l_{ijk}$. We will assume that $u_{ijk}$ and $l_{ijk}$ are rational numbers. Finally, let $S_\sigma^i$ be the set of all indices $(j, k)$ corresponding to the segments in all utility functions of agent $i$ under the given instance, i.e.,
\[  S_\sigma^i = \{ (j, k) \ | \ j \in G, \ k \in S_{ij} \} . \]

\begin{remark}
	\label{rem.ij}
	Throughout this paper, we will index elements of $A, G$ and $S_{ij}$ by $i, j$ and $k$, respectively. When the domain of $i, j$ or $k$ is not specified, especially in summations, it should be assumed to be $A, G$ and $S_{ij}$, respectively.
	\end{remark}

Program (\ref{eq.CP-SPLC}) is a convex program for SAD.

	\begin{maxi}
        {} {\sum_{i \in A}  {\log (u_i(x) - c_i)}}
			{\label{eq.CP-SPLC}}
		{}
        \addConstraint{\sum_{i \in A} \sum_{k \in S_{i j}} {x_{i j k}}}{\leq 1 }{\quad \forall j \in G}
        \addConstraint{\sum_{j \in G} \sum_{k \in S_{i j}} {x_{i j k}}}{\leq 1 }{\quad \forall i \in A}
		\addConstraint{x_{i j k}}{\leq l_{i j k}}{\quad \forall i \in A, \forall j \in G, \forall k \in S_{ij}}
        \addConstraint{x}{\geq 0.}
	\end{maxi}

The KKT conditions for program (\ref{eq.CP-LAD}) with non-negative dual variables $p$, $q$, and $h$ respectively are given below. By setting $c_i = 0$ for $i \in A$, we get KKT conditions for program (\ref{eq.CP-LF}) as well.

\begin{description}
\item [(1)]
    $ \forall j \in G: \ p_j > 0  \  \implies \ \sum_{i \in A} \sum_{k \in S_{i j}} {x_{ijk}}{=1 }$.
\item [(2)]
    $\forall i \in A: \ q_i > 0  \  \implies \ \sum_{j \in G} \sum_{k \in S_{i j}} {x_{i j k}}{=1 }$.
\item [(3)]
    $ \forall i \in A, \ \forall j \in G, \ \forall k \in S_{i j}: \ h_{i j k} > 0 \ \implies x_{i j k} = l_{i j k}$.
\item [(4)]
    $ \forall i \in A, \ \forall j \in G, \ \forall k \in S_{i j}: \ p_j + q_i + h_{i j k} \geq \frac{u_{i j k}}{u_i(x) - c_i}$.
\item [(5)]
    $ \forall i \in A, \ \forall j \in G, \ \forall k \in S_{i j k}:  x_{ij} > 0 \ \implies p_j + q_i + h_{i j k} = \frac{u_{i j k}}{u_i(x) - c_i}$.
\end{description}

As was the case for linear utilities, a $\frac{1}{2}$-approximate notion of equal-share fairness holds if $c_i = 0$ for all $i$.
For simplicity assume that there is some $\ell$ such that for all $i, j$, we have $S_{i j} = \{1, \ldots, \ell + 1\}$ and $\sum_{k = 1}^\ell {l_{i j k}} = 1$.
Note that due to the concavity of the utilities, the utility of player $i$ in the equal-share allocation is at most $\frac{1}{n} \sum_{j \in G} \sum_{k = 1}^\ell u_{i j k} l_{i j k}$.

\begin{lemma}\label{lem:opt_lb_splc}
    Let $x$ be an optimal solution to \eqref{eq.CP-LF} (or \eqref{eq.CP-LAD} with $c_i = 0$), then for all agents $i$ we have
    \[
        u_i(x) \geq \frac{1}{2 n} \sum_{j \in G} \sum_{k = 1}^\ell u_{i j k} l_{i j k}.
    \]
\end{lemma}

\begin{proof}
    The proof is similar to that of Lemma~\ref{lem:opt_lb}.
    By the KKT conditions, we know that $\frac{u_{i j k}}{u_i(x)} \leq p_j + q_i + h_{i j k}$ with equality if $x_{i j k} > 0$.

    Once again, one can then see that $q_i \leq 1$ and
    \begin{align*}
        \sum_{i \in A} \sum_{j \in G} \sum_{k = 1}^\ell l_{i j k} h_{i j k} + \sum_{j \in G} p_j + \sum_{i \in A} q_i &= \sum_{i \in A} \sum_{j \in G} \sum_{k = 1}^\ell x_{i j k} (p_j + q_i + h_{i j k}) \\
        &= n
    \end{align*}
    and therefore
    \begin{align*}
        u_i &= \max \left\{\frac{u_{i j k}}{p_j + q_i + h_{i j k }} \;\middle|\; j \in G, k \in [\ell]\right\} \\
            &\geq \sum_{j \in G} \sum_{k = 1}^\ell \frac{u_{i j k}}{p_j + q_i + h_{i j k}} \cdot \frac{l_{i j k} \cdot (p_j + q_i + h_{i j k})}{\sum_{j' \in G}p_{j'} + n q_i + \sum_{j' \in G} \sum_{k' = 1}^\ell l_{i j' k'} h_{i j' k'}} \\
            &\geq \frac{\sum_{j \in G} \sum_{k = 1}^\ell l_{i j k} u_{i j k}}{2 n}. \qedhere
    \end{align*}
\end{proof}

For the remaining models we will not get such strong equal-share fairness properties.
However, weakened versions are shown in Section~\ref{sec.cond_grad}.
In particular, see Lemma~\ref{lem:opt_lb_all}.

\subsection{Two-Sided Matching Markets}
\label{sec.2-models}

\textbf{1.)} Our two-sided matching market model consists of a set $A = \{1, 2, \ldots n\}$ of agents and a set $J = \{1, 2, \ldots, n\}$ of jobs. For uniformity, we have assumed that there is an equal number of agents and jobs, though the model can be easily enhanced and made more general. Our goal is to find an integral perfect matching between agents and jobs; however, we will relax this to finding a fractional perfect matching, $x$, followed by rounding as described above. We will explicitly define only the simplest case of two-sided markets; more general models follow along the same lines as one-sided markets.

Under the {\em linear Arrow-Debreu bargaining two-sided matching market}, abbreviated {\em 2AD}, the utility accrued by agent $i \in A$ under allocation $x$,
\[ u_i(x) = \sum_{j \in J} {u_{ij} x_{ij} }, \]
where $u_{ij}$ is the utility accrued by $i$ if she were assigned job $j$ integrally. Analogously,
the utility accrued by job $j \in J$ under allocation $x$,
\[  w_j(x) = \sum_{i \in A} {w_{ij} x_{ij} }, \]
where $w_{ij}$ is the utility accrued by $j$ if it were assigned to $i$ integrally.

In keeping with the axiom of symmetry under Nash bargaining, we will posit that the desires of agents and jobs are equally important and we are led to defining the feasible set in a $2n$ dimensional space, i.e., $\CN \subseteq \R_+^{2n}$. The first $n$ components of feasible point $v \in \CN$ represent the utilities derived by the $n$ agents, i.e., $u_i(x)$, and the last $n$ components the utilities derived by the $n$ jobs, i.e., $w_j(x)$, under a fractional perfect matching $x$. We seek the Nash bargaining point with respect to disagreement utilities $c_i$ for all agents $i \in G$ and $d_j$ for jobs $j \in J$.
If all $c_i$ and $d_j$ are 0, we call this model reduces to the {\em linear Fisher bargaining two-sided matching market} or \emph{2LF}.
A convex program of {\em 2AD} is given in (\ref{eq.2CP}).

Program (\ref{eq.2CP}) is a convex program for {\em 2AD}.
	\begin{maxi}
        {} {\sum_{i \in A}  {\log (u_i(x) - c_i)} \ + \ \sum_{j \in J} {\log (w_j(x) - d_j)}}
			{\label{eq.2CP}}
		{}
		\addConstraint{\sum_{i \in A} {x_{ij}}}{\leq 1 }{\quad \forall j \in J}
		\addConstraint{\sum_{j \in J} {x_{ij}}}{\leq 1 }{\quad \forall i \in A}
		\addConstraint{x}{\geq 0.}
	\end{maxi}

The KKT conditions for program (\ref{eq.2CP}) with non-negative dual variables $p$ and $q$ are:

\begin{description}
\item [(1)]
$\forall j \in J: \ p_j > 0 \ \implies \sum_{i \in A} {x_{ij}} = 1$.

\item [(2)]
$ \forall i \in A: \ q_i > 0 \implies \sum_{j \in J} {x_{ij}} = 1$.

\item [(3)]
    $ \forall i \in A, \ \forall j \in J: \ p_j + q_i \geq \frac{u_{ij}}{u_i(x) - c_i} + \frac{w_{ij}}{w_j(x) - d_j}$.

\item [(4)]
    $ \forall i \in A, \ \forall j \in J:  x_{ij} > 0 \implies p_j + q_i = \frac{u_{ij}}{u_i(x) - c_i} + \frac{w_{ij}}{w_j(x) - d_j}$.
\end{description}

\textbf{2.)} We may also extend this model even further by allowing SPLC utilities on both sides.
This defines the \emph{SPLC Arrow-Debreu bargaining two-sided matching market} or \emph{2SAD}.
In this setting, every edge $(i, j)$ comes with segments $k \in S_{i j}$ that have a length of $l_{i j k}$ as well as utilities $u_{i j k}$ for agent $i$ and $w_{i j k}$ for job $j$.
Note that we assume that the piecewise-linear, concave utility functions for $i$ and $j$ have the same breakpoints.
This is without loss of generality since we may always simply take the union of the breakpoints of both sides.

As before we have $u_i(x) = \sum_{j \in J} \sum_{k \in S_{i j}} u_{i j k} x_{i j k}$ and $w_j(x) = \sum_{i \in A} \sum_{k \in S_{i j}} w_{i j k} x_{i j k}$.
The convex programming formulation is then given by
\begin{maxi}
    {}
    {\sum_{i \in A} \log (u_i(x) - c_i) + \sum_{j \in J} \log (w_j(x) - d_j)}
    {}
    {\label{eq.2SAD}}
    \addConstraint{\sum_{i \in A} \sum_{k \in S_{i j}} x_{i j k}}{\leq 1}{\quad \forall j \in J}
    \addConstraint{\sum_{j \in J} \sum_{k \in S_{i j}} x_{i j k}}{\leq 1}{\quad \forall i \in A}
    \addConstraint{x_{i j k}}{\leq l_{i j k}}{\quad \forall i \in A, j \in J, k \in S_{i j}}
    \addConstraint{x}{\geq 0}
\end{maxi}

Its KKT conditions with non-negative dual variables $p$, $q$, and $h$ are:
\begin{description}
\item [(1)]
    $\forall j \in G: \ p_j > 0 \ \implies \sum_{i \in A} \sum_{k \in S_{i j}} {x_{ijk}} = 1$.

\item [(2)]
    $ \forall i \in A: \ q_i > 0 \implies \sum_{j \in G} \sum_{k \in S_{i j}} {x_{ijk}} = 1$.

\item [(3)]
    $ \forall i \in A, \ \forall j \in J, \ \forall k \in S_{i j}: \ h_{i j k} > 0 \implies x_{i j k} = l_{i j k}$.
\item [(4)]
    $ \forall i \in A, \ \forall j \in J, \ \forall k \in S_{i j}: \ p_j + q_i + h_{i j k} \geq \frac{u_{ij k}}{u_i(x) - c_i} + \frac{w_{ijk }}{w_j(x) - d_j}$.

\item [(5)]
    $ \forall i \in A, \ \forall j \in J, \ \forall k \in S_{i j}:  x_{ijk} > 0 \implies p_j + q_i + h_{i j k}  = \frac{u_{ijk}}{u_i(x) - c_i} + \frac{w_{ijk}}{w_j(x) - d_j}$.
\end{description}

\subsection{Non-Bipartite Matching Market}
\label{sec.Non-model}

Our non-bipartite matching market model consists of a set $A = \{1, 2, \ldots, n\}$ of agents and utilities for each pair $u_{ij} \in \RR_+$ which $i$ and $j$ accrue on getting matched to each other. Let $G = (A, E)$ be the complete graph on $n$ vertices, i.e., all edges are present. The vertices correspond to agents. A vector $x \in \R_+^{E}$ is said to be a {\em fractional matching in $G$} if it is a convex combination of (integral) matchings in $G$.
The goal is to find a Nash-bargaining solution with respect to disagreement utilities $c_i$ for all agents $i$.
This is a \emph{linear Arrow-Debreu bargaining non-bipartite matching market} or \emph{NBAD}.
If $c_i = 0$ for all agents, then we call this the \emph{linear Fisher bargaining non-bipartite matching market} or \emph{NBLF}.

Due to a classic result by Edmonds we have that $x$ is a fractional matching if and only if $x(\delta(i)) \leq 1$ for all $i$ and $x(E(B)) \leq \frac{|B| - 1}{2}$ for all sets $B \subseteq A$ of odd cardinality.
For ease of notation let us denote the collection of all such odd subsets $B$ of $A$ by $\mathcal{O}$.
This motivates the convex program:
\begin{maxi}
    {} {\sum_{i \in A}  {\log (u_i(x) - c_i)}}
        {\label{eq.NBAD}}
    {}
    \addConstraint{x(\delta(i))}{\leq 1 }{\quad \forall i \in A}
    \addConstraint{x(E(B))}{\leq \frac{|B| - 1}{2}}{\quad \forall B \in \mathcal{O}}
    \addConstraint{x}{\geq 0.}
\end{maxi}

The KKT conditions for program \eqref{eq.NBAD} with non-negative dual variables $p$ and $z$ are
\begin{description}
\item [(1)]
    $\forall i \in A: \ p_i > 0 \ \implies x(\delta(i)) = 1$.

\item [(2)]
    $ \forall B \in \mathcal{O}: \ z_B > 0 \implies x(E(B)) = \frac{|B| - 1}{2}$.

\item [(3)]
    $ \forall i, j \in A: \ p_i + p_j + \sum_{\{i, j\} \subseteq B \in \mathcal{O}} z_B \geq \frac{u_{ij}}{u_i(x) - c_i} + \frac{u_{ji}}{u_j(x) - c_j}$.
\item [(4)]
    $ \forall i, j \in A: \ x_{i j} > 0 \implies p_i + p_j + \sum_{\{i, j\} \subseteq B \in \mathcal{O}} z_B = \frac{u_{ij}}{u_i(x) - c_i} + \frac{u_{ji}}{u_j(x) - c_j}$.
\end{description}

We remark that as in the bipartite case, it is possible to take a fractional matching and (in a combinatorial, polynomial time manner), decompose it into a convex combination of most $n^2$ integral matchings.
This was shown by Padberg and Wolsey \cite{padberg1984fractional}; for a modern proof see \cite{vazirani2020extension}\footnote{The latter paper was written without knowledge of \cite{padberg1984fractional}, which has gone largely unknown in the research community.}.
This allows a similar rounding strategy if one desires integral allocations.

\subsection{Leontief Utilities}

In principle, one may also consider models with other classes of utilities.
Another common type of utility function in market models is that of \emph{Leontief utilities}.
In this setting, each agent $i \in A$ demands goods at some fixed ratio $a_{i j}$ for all $j \in G$.
For example, if the goods consist of items such as eggs, milk, flour, etc.\ and agent $i$ wishes to bake a cake, then they would require a very specific ratio of these ingredients as dictated by the recipe.
The utility of agent $i$ is then given simply by the amount of cake that they can make.

More formally, one may define a Nash-bargaining based matching market with Leontief utilities by the convex program
\begin{maxi*}
    {} {\sum_{i \in A}  {\log (u_i)}}
        {}
    {}
    \addConstraint{\sum_{j \in G}{u_i a_{i j}}}{\leq 1}{\quad \forall i \in A}
    \addConstraint{\sum_{i \in A}{u_i a_{i j}}}{\leq 1 }{\quad \forall j \in G}
    \addConstraint{x}{\geq 0.}
\end{maxi*}

However, these kinds of utilities does not really fit in with our other models for two main reasons:
\begin{enumerate}
    \item The ``unit demand constraints'' for each agent $i$ do not alter the problem in any way.
        One can simply add an extra good $j_i$ for each agent $i$ and set $a_{i j_i} = \sum_{j \in G}{a_{i j}}$.
        This implies that this model is really just a special case of a Fisher market with Leontief utilities which is well-studied.
    \item The feasible region of the polytope is not closely related to a matching polytope.
        Since each agent demands certain ratios of goods, a matching between goods and agents is generally a poor allocation.
        After all, getting one unit of eggs and nothing else is hardly useful if one wishes to bake a cake.
\end{enumerate}

For these reasons we do not consider models with Leontief utilities in this paper.

%% file: main2.tex
\section{Multiplicative Weights Update}\label{sec:mwua}

In this section we prove that MWU converges to an $\epsilon$-approximate Nash bargaining solution. The main result is given in the following theorem:
\begin{theorem}\label{thm:main} The following hold:

\begin{itemize}
\item MWUA (Algorithm \ref{alg:MWU}) computes an $\epsilon$-approximate Nash bargaining solution (see Definition \ref{def:approx}) of Program (\ref{eq.CP-LAD}), i.e.,\ model LiAD, after $O\left(\frac{n \log n}{\epsilon^2}\right)$ iterations and each iterate can be implemented in $O(n^2)$ time.
    \item MWUA (Algorithm \ref{alg:MWUsplc}) computes an $\epsilon$-approximate Nash bargaining solution of Program (\ref{eq.CP-SPLC}), i.e.,\ model SAD, after $O\left(\frac{n \log n}{\epsilon^2}\right)$ iterations and each iterate can be implemented in $O(l\cdot n^2 )$ time.
\end{itemize}
\end{theorem}

We provide the  analysis for the case of LiAD, i.e., linear utilities with endowments. The modified Lemmas so that the analysis carries over for piecewise linear utilities can be found in the end of the section.

\subsection{From optimization to feasibility}
Multiplicative Weights Update Algorithm (MWUA) \cite{Arora12} has found numerous applications in Game Theory, e.g., has been used  for computing Nash Equilibrium in zero-sum games and potential games or Correlated Equilibrium in general games. Another surprising application, is that MWUA can be used to find approximately feasible points for linear programs. Inspired by the latter, we use MWUA to find an (approximately) feasible solution to the feasibility program below (see Definition \ref{def:approx} for approximate feasibility):
\paragraph{Feasibility program.}

\begin{equation}\label{program:1}\tag{F-LiAD}
\begin{split}
&  \mathrm{CP}_i(p,q) \leq  u_i(x) \quad \forall i \in A,\\
&  \sum_{i\in A} x_{ij} \leq 1 \quad \forall j \in G,\\
&  \sum_{j\in G} x_{ij} \leq 1 \quad \forall i \in A,\\
& \sum_{j\in G} p_j +  \sum_{i\in A} q_i=n + \sum_{i \in A} c_i \min_{j \in G} \frac{p_j+q_i}{u_{ij}},\\
& x,p,q\geq 0
\end{split}
\end{equation}
where \begin{maxi*}
    {}
    {\sum_{j \in G} u_{ij}y_j}
    {}
    {\mathrm{CP}_i(p,q) \coloneqq}
    \addConstraint{\sum_{j \in G} (p_j+q_i)  y_j}{\leq 1+c_i \min_{j \in G} \frac{p_j+q_i}{u_{ij}}}
    \addConstraint{y}{\geq 0.}
\end{maxi*}

We are able to show that the Feasibility Program (\ref{program:1}) actually contains \textit{only} the optimal primal-dual solution of the convex program (\ref{eq.CP-LAD}), namely any feasible point of (\ref{program:1}) satisfies the KKT of (\ref{eq.CP-LAD}). This implies that as long as MWUA finds a solution that is approximately feasible for (\ref{program:1}), the solution will satisfy approximately the KKT conditions and hence will be approximately optimal. The first step towards the proof is to argue that if (\ref{program:1}) is feasible, that is there exist $(x^*,p^*,q^*)$ that satisfy the constraints in (\ref{program:1}), then and only then $(x^*,p^*,q^*)$ are the primal-dual variables that satisfy the KKT conditions. This is given in the following Lemma:

\begin{lemma}[Optimization to Feasibility]\label{lem:equivalence} Let $x^*$ be an optimal solution of Program (\ref{eq.CP-LAD}), then there exist $p^*,q^*$ so that $(x^*,p^*,q^*)$ is feasible for (\ref{program:1}). Conversely, assume that $(x^*,p^*,q^*)$ is a feasible solution for (\ref{program:1}), then $x^*$ is an optimal solution for Program (\ref{eq.CP-LAD}).
\end{lemma}

\begin{proof}
We first prove the direct, namely if $x^*$ is an optimal solution of Program (\ref{eq.CP-LAD}), then there exist $p^*,q^*$ so that $(x^*,p^*,q^*)$ is feasible for (\ref{program:1}).

Since $x^*$ is an optimal solution of Program (\ref{eq.CP-LAD}), there must be dual variables $w^*,z^*$ so that $(x^*,w^*,z^*)$ satisfies the KKT conditions, that is
\begin{equation}\label{eq:proofKKT}
\begin{split}
& \forall j \in G:  w^*_j \geq 0,  \forall i \in G:  z^*_i \geq 0,\\
& \forall j \in G: w^*_j = 0 \Rightarrow \sum_{i \in A} x^*_{ij}  \leq 1 \textrm{ and } w^*_j > 0 \Rightarrow \sum_{i \in A} x^*_{ij}  = 1,\\
& \forall i \in A: z^*_i = 0 \Rightarrow \sum_{j \in G} x^*_{ij}  \leq 1 \textrm{ and }z^*_i > 0 \Rightarrow \sum_{j \in G} x^*_{ij} = 1,\\
& \forall i \in A, \forall j \in G: x^*_{ij} = 0 \Rightarrow u_i(x^*) \geq \frac{u_{ij}}{w^*_j+z^*_i} + c_i,\\
& \forall i \in A, \forall j \in G: x^*_{ij} > 0 \Rightarrow u_i(x^*) = \frac{u_{ij}}{w^*_j+z^*_i} + c_i.
\end{split}
\end{equation}
We shall show that $(x^*,p^*,q^*)$ is feasible for (\ref{program:1}), by choosing $p^* = w^*$ and $q^*=z^*.$ Observe that by definition of $CP_i$ we get that
\begin{equation}
\begin{split}
\mathrm{CP}_i(p,q) &= \left(\max_{j \in G} \frac{u_{ij}}{p_j+q_i} \right) \times \left(1+c_i \min_{j\in G} \frac{p_j+q_i}{u_{ij}}\right)
\\& \left(\frac{1}{\min_{j \in G} \frac{p_j+q_i}{u_{ij}}} \right) \times \left(1+c_i \min_{j\in G} \frac{p_j+q_i}{u_{ij}}\right)
\\& = c_i + \frac{1}{\min_{j \in G} \frac{p_j+q_i}{u_{ij}}}
\\& = c_i + \max_{j \in G}\frac{u_{ij}}{p_j+q_i},
\end{split}
\end{equation}
and hence $\mathrm{CP}_i(p^*,q^*) = c_i + \max_{j \in G}\frac{u_{ij}}{p^*_j+q^*_i}.$
From KKT Equations (\ref{eq:proofKKT}) we also have that $u_i(x^*) \geq c_i + \frac{u_{ij}}{p^*_j+q^*_i}$ for all $j \in G$, i.e., $u_i(x^*) \geq c_i + \max_{j \in G}\frac{u_{ij}}{p_j^*+q_i^*}.$ We conclude that $u_i(x^*) \geq \mathrm{CP}_i(p^*,q^*).$

To finish the first part of the proof, i.e., that $(x^*,p^*,q^*)$ satisfies the constraints for (\ref{program:1}), it suffices to show that $$\sum_{j\in G} p^*_j+\sum_{i\in A} q^*_i =n + \sum_{i \in A} c_i \min_{j \in G} \frac{p^*_j+q^*_i}{u_{ij}}.$$

By KKT Equations (\ref{eq:proofKKT}), one has that
\begin{equation*}
\begin{split}
\sum_{j\in G} p^*_j+ \sum_{i\in A} q^*_i &= \sum_{j\in G} p^*_j \sum_{i\in A}x^*_{ij}+ \sum_{i\in A} q^*_i \sum_{j\in G}x^*_{ij} \\
& = \sum_{i \in A}\sum_{j \in G} x^*_{ij} (p^*_j+q^*_i)\\
& = \sum_{i \in A}\sum_{j \in G} x^*_{ij} \frac{u_{ij}}{u_i(x^*)-c_i}\\
& = \sum_{i \in A} \frac{1}{u_i(x^*)-c_i} \sum_{j \in G} u_{ij}x^*_{ij} \\
& = \sum_{i \in A} \frac{u_i(x^*)}{u_i(x^*)-c_i}  \\
& = \sum_{i \in A} 1 + \frac{c_i}{u_i(x^*)-c_i}  \\
& = n+ \sum_{i \in A} \frac{c_i}{u_i(x^*)-c_i} = n+ \sum_{i \in A} c_i\min_{j \in G} \frac{p^*_j+q^*_i}{u_{ij}}.
\end{split}
\end{equation*}
The first part of the proof is complete.

For the converse direction, assuming that $(x^*,p^*,q^*)$ satisfies the constraints of (\ref{program:1}), we shall show that $x^*$ is a maximizer for Program (\ref{eq.CP-LAD}). By assumption, one has $u_i(x^*) \geq \mathrm{CP}_i(p^*,q^*)$ for all $i \in A$, therefore since $\mathrm{CP}_i(p^*,q^*)$ is defined to be the maximum utility $i$ can get with prices $p^*,q^*$, $i$ must spend at least all his budget, that is for all agents $i \in A$
\begin{equation}\label{eq:overbudget}
\sum_{j \in G} \left( p^*_j+q^*_i\right)x^*_{ij} \geq 1+ c_i \min_{j \in G}  \frac{p^*_{j}+q^*_i}{u_{ij}}.
\end{equation}
By summing the above inequality (\ref{eq:overbudget}) for all $i \in A$ we conclude that
\begin{equation}\label{eq:overbudgetall}
\sum_{i\in A}\sum_{j \in G} \left( p^*_j+q^*_i\right)x^*_{ij}\geq n+ \sum_{i \in A}c_i \min_{j \in G}  \frac{p^*_{j}+q^*_i}{u_{ij}}.
\end{equation}
Moreover, using (\ref{program:1}) it holds
\begin{equation}\label{eq:budget2}
\begin{split}
n + \sum_{i \in A}c_i \min_{j \in G}  \frac{p^*_{j}+q^*_i}{u_{ij}} &= \sum_{j\in G} p^*_j+\sum_{i\in A} q^*_i \\
& \geq \sum_{j\in G} p^*_j \left(\sum_{i \in A} x^*_{ij}\right)+\sum_{i\in A} q^*_i \left(\sum_{j \in G} x^*_{ij}\right)\\
& = \sum_{i \in A} \sum_{j \in G} \left( p^*_j+q^*_i\right)x^*_{ij},
\end{split}
\end{equation}
where the inequality comes from the fact $\sum_{i \in A} x^*_{ij} \leq 1$ for all $j \in G$ and $\sum_{j \in G} x^*_{ij}\leq 1$ for all $i\in A.$

Using (\ref{eq:overbudgetall}) and (\ref{eq:budget2}) it follows that we actually have equality in (\ref{eq:overbudgetall}), that is

\begin{equation}
\begin{split}
\sum_{i\in A}\sum_{j \in G} \left(p^*_j+q^*_i\right) x^*_{ij} &= n+ \sum_{i \in A}c_i \min_{j \in G}  \frac{p^*_{j}+q^*_i}{u_{ij}},\\\textrm{and hence (\ref{eq:overbudget}) is also  equality, i.e.,}&\;\;\;\;\;\;\;\;\;\;\;\;\;\;\;\;\;\;\;\;\;\;\;\;\;\;\;\;\;\;\;\;\;\;\;\;\;\;\;\;\;\;\;\;\;\;\;\;\;\;\;\;\;\;\;\;\;\;\;\;\;\;\;\;\;\;\;\;\;\;\;\;\;\;\;\;\;\;\;\;\;\;\;\;\;\;\;\;\;\;\;\;\;\;\;\;\;\;\;\;\;\;\;\;
\\
\sum_{j \in G} \left(p^*_j+q^*_{i}\right)x^*_{ij} &= 1+ c_i \min_{j \in G}  \frac{p^*_{j}+q^*_i}{u_{ij}} \textrm{ for all }i\in A.
\end{split}
\end{equation}
Additionally, the inequality in (\ref{eq:budget2}) indicates that whenever $p^*_j>0$ we get  $\sum_{i \in A} x^*_{ij}=1$ and whenever $q^*_i>0$ we get  $\sum_{j \in G} x^*_{ij}=1$. Thus, $x^*$ is feasible for $\mathrm{CP}_i(p^*,q^*)$ and since $\mathrm{CP}_i(p^*,q^*) \leq u_i(x^*)$, $x^*$ must be a maximizer of $\mathrm{CP}_i(p^*,q^*)$ for all $i\in A.$

\smallskip
Finally, since $x^*$ is the maximizer of $\mathrm{CP}_i(p^*,q^*)$ it should hold for all $i \in A$
\begin{equation*}
\begin{split}
u_i(x^*) &= \mathrm{CP}_i (p^*,q^*)\\ &=\left(\max_{j \in G} \frac{u_{ij}}{p^*_{j}+q^*_i} \right) \times \left(1+c_i \min_{j \in G} \frac{p^*_j+q^*_i}{u_{ij}}\right)
\\&= c_i + \max_{j \in G}\frac{u_{ij}}{p^*_{j}+q^*_i}
\\&= c_i + \frac{u_{ij}}{p^*_{j}+q^*_i} \textrm{ for all }j\in G \textrm{ s.t. }x^*_{ij}>0.
\end{split}
\end{equation*}

\noindent We conclude that $u_{i}(x^*) \geq c_i + \frac{u_{ij}}{p^*_j+q^*_i},$ with equality only if $x^*_{ij}>0$. The proof is complete because we showed that the KKT conditions (\ref{eq:proofKKT}) are satisfied.
\end{proof}

A corollary that can be derived from Lemma \ref{lem:equivalence} is that an approximate feasible solution of (\ref{program:1}) is an approximately optimal Nash bargaining solution of Program (\ref{eq.CP-LAD}). Before we state the corollary, we provide the formal definition of an $\epsilon$-approximate feasible solution to (\ref{program:1}) and $\epsilon$-approximate Nash bargaining solution.

\begin{definition}[Approximately feasible]\label{def:approx} A point $(\tilde{x},\tilde{p}, \tilde{q})$ is called $\epsilon$-approximate feasible of (\ref{program:1}), if it satisfies non-negativity, price and max-utility constaints but might violate allocation constraints by an additive $\epsilon$; that is we have $\sum_{i\in {A}} \tilde{x}_{i j} \leq 1+\epsilon,$ for all $j \in G$ and $\sum_{j\in {G}} \tilde{x}_{i j} \leq 1+\epsilon,$ for all $i \in A .$ Moreover, we define an $\epsilon$-approximate Nash bargaining solution for program (\ref{eq.CP-LAD}) to be a feasible allocation $x$ such that for each agent $i$, if $i$ changes his allocation from $x_i$ to some feasible allocation $x'_i$, he cannot gain more than an additive $\epsilon$ in his utility.
\end{definition}

\begin{corollary}\label{cor:approximate}Let $(\tilde{x},\tilde{p},\tilde{q})$ be a $\epsilon$-approximate feasible point of (\ref{program:1}).  Then $\frac{1}{1+\epsilon}\cdot \tilde{x}$ is a $O(\epsilon)$-approximate Nash bargaining solution of Program (\ref{eq.CP-LAD}).
\end{corollary}

\begin{proof}
Let $\overline{x},\overline{p},\overline{q}$ be the time averages of allocations and prices respectively that the Algorithm \ref{alg:MWU} returns and let $x^*$ denote the optimal solution of Program (\ref{eq.CP-LAD}). Using the fact that $\mathrm{CP}_i(p,q)$ is convex (Claim \ref{lem:cpconvex}), we showed Inequality (\ref{eq:needone}), that is $\mathrm{CP}_i(\overline{p},\overline{q}) \leq u_i(\overline{x})$. As a result, it must hold for each agent $i$
\[
 \log(u_i(x^*)-c_i) \leq  \log(u_i(\overline{x})-c_i).
\]
Moreover, it holds $\ln (1+\epsilon) \leq \epsilon$, hence
\begin{equation*}
\begin{split}
 \log\left( u_i(\overline{x})-c_i\right) - \log\left(\frac{1}{\epsilon+1}\cdot u_i(\overline{x})-c_i\right) &=\log \left(1 + \frac{\frac{\epsilon}{1+\epsilon}u_{i}(\overline{x})}{\frac{1}{1+\epsilon} \cdot u_{i}(\overline{x})-c_i}\right)
\\&\leq   \frac{\epsilon \cdot u_{i}(\overline{x})}{u_{i}(\overline{x}) - (1+\epsilon)c_i}
\end{split}
\end{equation*}
which is $O(\epsilon )$.
\end{proof}

\begin{algorithm}[h!]
  \caption{Multiplicative Weight Update}


  \begin{algorithmic}

  \State $\tilde{p}^{(0)} \leftarrow\textbf{1}, \tilde{q}^{(0)} \leftarrow\textbf{1}$ \/\/ initialization

  \For{$t=1$ to $T$}

     \State \textbf{Rescale} $\tilde{p}^{(t)}, \tilde{q}^{(t)}$ so that $\sum_{j \in G} p_j^{(t)}+ \sum_{i \in A} q_i^{(t)}= n + \sum_{i \in A} c_i \min_{j \in G} \frac{p_j^{(t)}+q_i^{(t)}}{u_{ij}}.$

     \For{$i=1$ to $n$}

     \State  $x_i^{(t)} \leftarrow \arg\max_y \left\{\sum_{j \in G} u_{ij}y_j : \sum_{j \in G}\left(p^{(t)}_j+q_i^{(t)}\right) y_j \leq 1+c_i \min_{j \in G} \frac{p_j^{(t)}+q_i^{(t)}}{u_{ij}}\right\}.$
     \EndFor

     \For{$j=1$ to $n$}
     \State $d_j^{(t)} \leftarrow \sum_{i\in A} x_{ij}^{(t)}$
     \EndFor

     \For{$i=1$ to $n$}
     \State $h_i^{(t)} \leftarrow \sum_{j \in G} x_{ij}^{(t)}$
     \EndFor

     \State $\sigma^{(t)} \leftarrow \min \left(\frac{1}{\max_{j\in G} d_j^{(t)}} , \frac{1}{\max_{i\in A} h_i^{(t)}}\right)$

     \For{$j=1$ to $n$}
     \State $\tilde{p}_j^{(t+1)} \leftarrow \tilde{p}_j^{(t)} (1+\epsilon \sigma^{(t)}d_j^{(t)})$
     \EndFor

     \For{$i=1$ to $n$}
     \State $\tilde{q}_i^{(t+1)} \leftarrow \tilde{q}_i^{(t)} (1+\epsilon \sigma^{(t)}h_i^{(t)})$
     \EndFor
  \EndFor

  \State \textbf{return } $\overline{x} \leftarrow \frac{\sum_{t=1}^T \sigma^{(t)}x^{(t)}}{\sum_{t=1}^T \sigma^{(t)}}$ (matrix), $\overline{p} \leftarrow \frac{\sum_{t=1}^T \sigma^{(t)}p^{(t)}}{\sum_{t=1}^T \sigma^{(t)}}$, $\overline{q} \leftarrow \frac{\sum_{t=1}^T \sigma^{(t)}q^{(t)}}{\sum_{t=1}^T \sigma^{(t)}}$.
  \end{algorithmic}
  \label{alg:MWU}
\end{algorithm}

\subsection{MWUA Analysis for LiAD via (\ref{program:1})}
Since we showed equivalence (Lemma \ref{lem:equivalence}) between Program (\ref{eq.CP-LAD}) and (\ref{program:1}), we will apply MWUA on the latter (see Algorithm \ref{alg:MWU}). The rest of the section will be focusing on proving convergence and bounding the rate of convergence of MWUA. We shall show that after $O\left(\frac{n \log n}{\epsilon^2}\right)$ iterations, MWUA will reach a point $(\tilde{x},\tilde{p},\tilde{q})$ that is an $\epsilon$-approximate feasible solution for (\ref{program:1}) and hence $\frac{1}{1+\epsilon}\cdot \tilde{x}$ will be an $O(\epsilon)$-approximate optimal solution for each agent in Program (\ref{eq.CP-LAD}), as stated in the aforementioned Corollary \ref{cor:approximate}. 

Before we proceed with the analysis of Algorithm \ref{alg:MWU}, we need to argue that the rescaling step in the For-loop is well-defined, i.e., we can always rescale $\tilde{p}^{(t)},\tilde{q}^{(t)}$ to $p^{(t)},q^{(t)}$ so that $\sum_{j \in G} p_j^{(t)}+ \sum_{i \in A} q_i^{(t)}= n + \sum_{i \in A} c_i \min_{j \in G} \frac{p_j^{(t)}+q_i^{(t)}}{u_{ij}}.$ This is captured by the following lemma:
\begin{lemma}[Rescaling]\label{lem:scale} Assume (\ref{program:1}) is feasible. Given $\tilde{p},\tilde{q} \geq \textbf{1}$, we can always rescale them to $p,q$ so that $\sum_{j \in G} p_j+ \sum_{i \in A} q_i= n + \sum_{i \in A} c_i \min_{j \in G} \frac{p_j+q_i}{u_{ij}}.$
\end{lemma}
\begin{proof} Let $\sum_{j \in G} \tilde{p}_j+ \sum_{i \in A} \tilde{q}_i - \sum_{i \in A} c_i \min_{j \in G} \frac{\tilde{p}_j+\tilde{q}_i}{u_{ij}} = \lambda$. As long as $\lambda>0$ then we can set $p = \frac{n}{\lambda}\cdot \tilde{p}, q = \frac{n}{\lambda}\cdot \tilde{q}$ and the lemma would follow. Therefore it suffices to show $\lambda$ is positive. Using the assumption (\ref{program:1}) is feasible we get that there exists $x^*$ with $u_i(x^*)>c_i$ for all $i \in A$, thus 
$\max_{j \in G} u_{ij} \geq u_i(x^*)>c_i.$ We conclude  
\begin{align*}
\sum_{i \in A} c_i \min_{j \in G} \frac{\tilde{p}_j+\tilde{q}_i}{u_{ij}} &\leq \sum_{i \in A} c_i \min_{j \in G}\frac{\tilde{p}_j+\tilde{q}_i}{u_{i}(x^*)}\\
&\leq \frac{1}{n} \sum_{i \in A} c_i  \sum_{j \in G}\frac{\tilde{p}_j+\tilde{q}_i}{u_{i}(x^*)}\\&< \frac{1}{n} \sum_{i \in A}   \sum_{j \in G}\tilde{p}_j+\tilde{q}_i = \sum_{j \in G}\tilde{p}_j+\sum_{i \in A}\tilde{q}_i, 
\end{align*}
from which follows $\lambda>0.$
\end{proof}
\paragraph{ Potential function.} To bound the rate of convergence of MWUA and show that it actually reaches an $O(\epsilon)$-approximate feasible point of (\ref{program:1}), we use a potential function argument, an idea that is quite common in online learning literature \cite{Arora12}.
The potential function at iterate $t$ is defined to be (sum of prices at iterate $t$)
\begin{equation}\label{eq:potential}
\Phi(t) = \sum_{j \in G} \tilde{p}_j^{(t)} + \sum_{i \in A} \tilde{q}_i^{(t)}.
\end{equation}

 We need to find upper and lower bounds on $\Phi(T)$ which will enable us to bound the convergence rate of MWUA ($T$ denotes the number of iterations of MWUA). The first lemma gives an upper bound on the potential function $\Phi$.
\begin{lemma}[Upper bound on $\Phi(T)$]\label{lem:upper} The following holds:
\begin{equation}\label{eq:upper}
\Phi(T) \leq 2n \cdot \textrm{exp}\left(\epsilon \sum_{t=1}^{T-1} \sigma^{(t)}\right),
\end{equation}
where $\sigma^{(t)} := \min\left(\frac{1}{\max_{j \in G} \sum_{i \in A} x^{(t)}_{ij}},\frac{1}{\max_{i \in A} \sum_{j \in G} x^{(t)}_{ij}}\right)$ \footnote{$\sigma^{(t)}$ is essentially the renormalization parameter so that the ``utility" of each constraint/expert in MWU is bounded by one.} (see Algorithm \ref{alg:MWU}).
\end{lemma}

\begin{proof} Fix a time index $t$. We get that $\frac{\Phi(t+1)}{\Phi(t)}$ is equal to 

\begin{align*}
 &= \frac{\sum_{j \in G} \tilde{p}_j^{(t+1)}+\sum_{i \in A} \tilde{q}_i^{(t+1)}}{\Phi(t)}
 \;\;\;\;\;\;\;\;\;\;\;\;\;\;\;\;\;\;\;\;\;\;\;\;\;\;\;\;\;\;\;\;\;\;\;\;\;\;\textrm{(definition of } \Phi(t+1))\\
&= \frac{\sum_{j \in G} (1+\epsilon \sigma^{(t)}d_j^{(t)})\tilde{p}_j^{(t)}+\sum_{i \in A} (1+\epsilon \sigma^{(t)}h_i^{(t)})\tilde{q}_i^{(t)} }{\Phi(t)}
\;\;\textrm{ (def. of Algorithm \ref{alg:MWU})}
\\&= 1+\epsilon \sigma^{(t)} \frac{\sum_{j \in G}d_j^{(t)}\tilde{p}_j^{(t)}+ \sum_{i \in A}h_i^{(t)}\tilde{q}_i^{(t)}}{\Phi(t)}
\end{align*}
\begin{align*}
& = 1+ \epsilon \sigma^{(t)} \frac{\sum_{j \in G} \sum_{i\in A}  x_{ij}^{(t)}p_j^{(t)} + \sum_{i \in A} \sum_{j\in G}  x_{ij}^{(t)}q_i^{(t)}}{a^{(t)}\Phi(t)}
\;\textrm{(rescale }a^{(t)} \textrm{ and def. of }d_j^{(t)},h_i^{(t)})
\\&= 1+ \epsilon \sigma^{(t)} \frac{\sum_{i\in A} \sum_{j \in G} \left(p_j^{(t)}+q_i^{(t)}\right)x_{ij}^{(t)} }{a^{(t)}\Phi(t)}
\;\;\;\;\;\;\;\;\;\;\;\;\;\;\;\;\;\;\;\;\textrm{(change order of sum)}
\\&\leq 1+ \epsilon \sigma^{(t)} \frac{\sum_{i\in A} 1+c_i \min_{j \in G} \frac{p^{(t)}_j+q^{(t)}_i}{u_{ij}} }{a^{(t)}\Phi(t)} \;\;\;\;\;\;\;\;\;\;\;\;\;\;\;\;\;(x_i^{(t)} \textrm{ is argmax in Alg. \ref{alg:MWU}})
\\&= 1+ \epsilon \sigma^{(t)} \frac{n+\sum_{i\in A} c_i \min_{j \in G} \frac{p^{(t)}_j+q^{(t)}_i}{u_{ij}}  }{a^{(t)}\Phi(t)} \\&=1 + \epsilon \sigma^{(t)} \frac{\sum_{j \in G} p^{(t)}_j +\sum_{i \in A} q^{(t)}_i }{a^{(t)}\Phi(t)} \;\;\;\textrm{(definition of }p^{(t)}_j,q^{(t)}_i
\;\;\;\textrm{ in Algorithm \ref{alg:MWU})}
\\&= 1+ \epsilon \sigma^{(t)} \leq 1+\textrm{exp}\left(\epsilon \sigma^{(t)}\right).
\end{align*}

Multiplying telescopically we have that \[\Phi(T) \leq \Phi(1) \cdot \textrm{exp}\left(\epsilon \sum_{t=1}^{T-1}\sigma^{(t)}\right) = 2n \cdot\textrm{exp}\left(\epsilon \sum_{t=1}^{T-1}\sigma^{(t)}\right)\] and the claim follows.
\end{proof}

In what follows, we prove a lower bound on the potential function $\Phi$.
\begin{lemma}[Lower bound on $\Phi(T)$]\label{lem:lower} The following holds:
\begin{equation}
\Phi(T) \geq \textrm{exp}\left(\epsilon(1-\epsilon)    \max\left(\max_{j \in G} \sum_{t=1}^{T} \sigma^{(t)}d_j^{(t)},\max_{i \in A} \sum_{t=1}^{T} \sigma^{(t)}h_i^{(t)}\right)\right),
\end{equation}
where $d_j^{(t)} := \sum_{i \in A} x^{(t)}_{ij}$ and $h_i^{(t)} := \sum_{j \in G} x^{(t)}_{ij}$ (see Algorithm \ref{alg:MWU}).
\end{lemma}
\begin{proof}
We will need the following straightforward auxiliary fact:
\\\textbf{Claim.} $e^{\mu(1-y)} \leq 1+\mu$ for  $0<\mu \leq y<1.$ By definition of $\Phi$ we get

\begin{align*}
\Phi(T) &= \sum_{j \in G} \tilde{p}^{(T)}_j+\sum_{i \in A} \tilde{q}^{(T)}_i
\\&\geq \sum_{j \in G} \prod_{t=1}^{T-1}(1+\epsilon \sigma^{(t)}d_j^{(t)})+\sum_{i \in A} \prod_{t=1}^{T-1}(1+\epsilon \sigma^{(t)}h_i^{(t)})\;\textrm{(definition of Algorithm \ref{alg:MWU})}
\\& \geq \sum_{j \in G} \textrm{exp}\left(\epsilon(1-\epsilon)\sum_{t=1}^{T-1}\sigma^{(t)}d_j^{(t)}\right) + \sum_{i \in A} \textrm{exp}\left(\epsilon(1-\epsilon)\sum_{t=1}^{T-1}\sigma^{(t)}h_i^{(t)}\right)\;\;\;\;\;\;\;\;\;\;\;\textrm{(use of Claim)}
\\& \geq \max\left(\max_{j \in G} \textrm{exp}\left(\epsilon(1-\epsilon)\sum_{t=1}^{T-1}\sigma^{(t)}d_j^{(t)}\right),\max_{i \in A} \textrm{exp}\left(\epsilon(1-\epsilon)\sum_{t=1}^{T-1}\sigma^{(t)}h_i^{(t)}\right)\right)
\\& =  \textrm{exp}\left(\epsilon(1-\epsilon)\max\left(\max_{j \in G} \sum_{t=1}^{T} \sigma^{(t)}d_j^{(t)},\max_{i \in A} \sum_{t=1}^{T} \sigma^{(t)}h_i^{(t)}\right)\right).
\end{align*}
\end{proof}

Combining Lemmas \ref{lem:upper} and \ref{lem:lower} we conclude that

\begin{equation}\label{eq:combined}
\epsilon(1-\epsilon) \max\left(\max_{j \in G} \sum_{t=1}^{T} \sigma^{(t)}d_j^{(t)},\max_{i \in A} \sum_{t=1}^{T} \sigma^{(t)}h_i^{(t)}\right)\leq \ln \Phi(T) \leq \ln (2n) + \epsilon \sum_{t=1}^{T-1}\sigma^{(t)}.
\end{equation}

\paragraph{ Proof of part (a) of Theorem \ref{thm:main}.}
Let $\overline{x} = \frac{\sum_{t=1}^T \sigma^{(t)}x^{(t)}}{\sum_{t=1}^T \sigma^{(t)}}$ and $\overline{p} = \frac{\sum_{t=1}^T \sigma^{(t)}p^{(t)}}{\sum_{t=1}^T \sigma^{(t)}}$ and $\overline{q} = \frac{\sum_{t=1}^T \sigma^{(t)}q^{(t)}}{\sum_{t=1}^T \sigma^{(t)}}$ (average among all allocations, prices at all times, the output of MWUA \ref{alg:MWU}). Observe that $\overline{x}, \overline{p}$ are non-negative and satisfy the prices constraints (because of the rescaling step). We will show, as long as $T$ is chosen to be $\Theta\left(\frac{2n\log 2n}{\epsilon^2}\right)$, that
\begin{equation}\label{eq:needone}
\mathrm{CP}_i(\overline{p},\overline{q}) \leq u_i (\overline{x}) \textrm{ for all }i \in A,
\end{equation}
and moreover
 \begin{equation}\label{eq:needtwo}
 \begin{split}
\sum_{i\in A} \overline{x}_{i j} \leq 1 + O(\epsilon) \textrm{ for all }j \in G \textrm{ and }
\sum_{j\in G} \overline{x}_{i j} \leq 1 + O(\epsilon) \textrm{ for all }i \in A.
\end{split}
\end{equation}

By showing (\ref{eq:needone}) and (\ref{eq:needtwo}) we will get that $(\overline{x},\overline{p},\overline{q})$ is an $O(\epsilon)$-approximate feasible of (\ref{program:1}) and thus from Corollary \ref{cor:approximate} $\frac{1}{1+O(\epsilon)} \cdot \overline{x}$ will be $O(\epsilon)$-approximate optimal solution for each agent in Program (\ref{eq.CP-LAD}).

\paragraph{ Showing (\ref{eq:needone}).} To show that $\mathrm{CP}_i(\overline{p},\overline{q}) \leq u_i (\overline{x})$ for all $i \in A$, observe that
\begin{equation}\label{eq:aux}
\mathrm{CP}_i(p^{(t)},q^{(t)}) \leq u_i (x^{(t)}) \;\forall i \in A \textrm{ (by choice of }x^{(t)} \textrm{ in Algorithm \ref{alg:MWU})}.
\end{equation}

 It is not hard to see that $\mathrm{CP}_i(p,q)$ is a convex function for all $i \in A$:
\begin{claim}[Convexity of $CP_i$]\label{lem:cpconvex}
$\mathrm{CP}_i(p,q)$ is convex for $p > 0$.
\end{claim}
\begin{proof} It is a well-known fact that the maximum of convex functions is convex. Moreover the function $f(x_1,...,x_k) = \frac{1}{\sum_{s=1}^k a_s x_s}$ with $a_s$ non-negative is convex, therefore $\mathrm{CP}_i(p,q) = c_i + \max_{j \in G} \left\{\frac{u_{ij}}{p_j+q_i}\right\}$ is also convex.
\end{proof}

Using Claim \ref{lem:cpconvex} and Jensen's inequality, it follows that \[CP_i(\overline{p},\overline{q}) \stackrel{\ref{lem:cpconvex}}{\leq} \frac{\sum_{t=1}^T \sigma^{(t)} CP_i(p^{(t)},q^{(t)})}{\sum_{t=1}^T \sigma^{(t)}} \stackrel{(\ref{eq:aux})}{\leq} \frac{\sum_{t=1}^T \sigma^{(t)} u_i(x^{(t)}_i)}{\sum_{t=1}^T \sigma^{(t)}}
= u_{i}(\overline{x}), \]
i.e., inequality (\ref{eq:needone}) is proved.

\paragraph{Showing (\ref{eq:needtwo}).} For the rest of the section, we will focus on showing (\ref{eq:needtwo}), that is (for definition of $\sigma^{(t)}, d^{(t)}, h^{(t)}$ see Algorithm \ref{alg:MWU}) we need to show
\[\sum_{t=1}^T \sigma^{(t)}d^{(t)}_j  := \sum_{i\in A} \overline{x}_{i j} \leq 1 + O(\epsilon) \textrm{ for all }j \in G\] and
\[\sum_{t=1}^T \sigma^{(t)}h^{(t)}_i  := \sum_{j \in G} \overline{x}_{i j} \leq 1 + O(\epsilon) \textrm{ for all }i \in A\]
when $T$ is $\Theta(\frac{2n\log 2n}{\epsilon^2}).$ We shall show that
\begin{equation*}
\max\left(\max_{j \in G}\sum_{i\in A}   \overline{x}_{i j},\max_{i \in A}\sum_{j\in G}   \overline{x}_{i j}\right) \leq 1 + O(\epsilon),
\end{equation*}
and (\ref{eq:needtwo}) follows. Using Inequality (\ref{eq:combined}) we get
\begin{align*}
\max\left(\max_{j \in G}\sum_{i\in A}   \overline{x}_{i j},\max_{i \in A}\sum_{j\in G}   \overline{x}_{i j}\right) &= \max\left(\max_{j \in G} \sum_{t=1}^{T} \sigma^{(t)}d_j^{(t)},\max_{i \in A} \sum_{t=1}^{T} \sigma^{(t)}h_i^{(t)}\right)
\\&\leq \frac{\ln (2n) + \epsilon \sum_{t=1}^{T}\sigma^{(t)}}{\epsilon(1-\epsilon) \sum_{t=1}^{T}\sigma^{(t)}}
\\& = \frac{1}{1-\epsilon} + \frac{\ln (2n) }{\epsilon(1-\epsilon) \sum_{t=1}^{T}\sigma^{(t)}}
\end{align*}

Observe that at every iteration $t$, there exists a $j \in G$ or $i \in A$ so that $\tilde{p}^{(t)}_j$ or $\tilde{q}^{(t)}_i$ increases by a factor of $(1+\epsilon)$ (for iteration $t$, it should be the $\arg\max_{i,j} \{ d^{(t)}_j,h^{(t)}_i\}$). Hence after $T$ iterations it follows that (Pigeonhole Principle)

\begin{equation}\label{eq:lbound}\max\left(\max_{j \in G} \tilde{p}_{j}^{(T+1)},\max_{i \in A} \tilde{q}_{i}^{(T+1)}\right) \geq (1+\epsilon)^{T/2n} \approx (2n)^{1/\epsilon}, \end{equation}
where the last approximation holds by setting $T = \frac{2n\log 2n}{\epsilon^2}.$

\noindent Finally it holds (using the fact that $e^x \geq x+1$ for all $x$)
\begin{align*}
\max &\left(\max_{j \in G} \frac{\sum_{t=1}^{T}\sigma^{(t)}d_j^{(t)}}{\sum_{t=1}^T \sigma^{(t)}}
 ,\max_{i \in A} \frac{\sum_{t=1}^{T}\sigma^{(t)}h_i^{(t)}}{\sum_{t=1}^T \sigma^{(t)}}\right)\\&= \max\left(\max_{j \in G}\frac{ \ln \prod_{t=1}^T \textrm{exp}\left( \epsilon\sigma^{(t)}d_j^{(t)}\right)}{\epsilon\sum_{t=1}^T \sigma^{(t)}},\max_{i \in A}\frac{ \ln \prod_{t=1}^T \textrm{exp}\left( \epsilon\sigma^{(t)}h_i^{(t)}\right)}{\epsilon\sum_{t=1}^T \sigma^{(t)}}\right)
\\&\geq \max\left(\max_{j \in G}\frac{ \ln \prod_{t=1}^T (1+ \epsilon\sigma^{(t)}d_j^{(t)})}{\epsilon\sum_{t=1}^T \sigma^{(t)}},\max_{i \in A}\frac{ \ln \prod_{t=1}^T (1+ \epsilon\sigma^{(t)}h_i^{(t)})}{\epsilon\sum_{t=1}^T \sigma^{(t)}}\right)
\\& = \max\left(\max_{j \in G}\frac{\ln  \tilde{p}_j^{(T+1)}}{\epsilon \sum_{t=1}^T \sigma^{(t)}},\max_{i \in A}\frac{\ln \tilde{q}_i^{(T+1)}}{\epsilon \sum_{t=1}^T \sigma^{(t)}}\right)
\\& = \ln\left(\max\left(\max_{j \in G}\frac{  \tilde{p}_j^{(T+1)}}{\epsilon \sum_{t=1}^T \sigma^{(t)}},\max_{i \in A}\frac{  \tilde{q}_i^{(T+1)}}{\epsilon \sum_{t=1}^T \sigma^{(t)}}\right)\right) \;(\ln \textrm{ is increasing) }\\&\approx \frac{\ln 2n}{\epsilon^2 \sum_{t=1}^T \sigma^{(t)} } \textrm{ using (\ref{eq:lbound}).}
\end{align*}

Hence,

\[\max\left(\max_{j \in G} \sum_{i \in A} \overline{x}_{ij},\max_{i \in A} \sum_{j \in G} \overline{x}_{ij} \right) \leq \frac{1}{1-\epsilon} + \frac{\epsilon}{1-\epsilon}\cdot \max\left(\max_{j \in G} \sum_{i \in A} \overline{x}_{ij},\max_{i \in A} \sum_{j \in G} \overline{x}_{ij} \right),\] or equivalently
\[\max\left(\max_{j \in G} \sum_{i \in A} \overline{x}_{ij},\max_{i \in A} \sum_{j \in G} \overline{x}_{ij} \right)\leq \frac{1}{1-2\epsilon} = 1 + O(\epsilon).\]

Having shown (\ref{eq:needone}) and (\ref{eq:needtwo}), we conclude that $(\overline{x},\overline{p},\overline{q})$ is an $O(\epsilon)$-approximate feasible for \ref{program:1} and the proof of part (a) of Theorem \ref{thm:main} is complete.
\subsection{SPLC utilities (SAD)}
The previous approach can be generalized even for model SAD with piecewise linear concave utilities. We present the parts of the proof that need modification and the modified MWUA Algorithm. In a nutshell, we modify the price (dual variables) constraints to account for the extra constraints on the slopes. The feasibility program (analogously with \ref{program:1}) will have the form:
\paragraph{Feasibility program.}

\begin{equation}\label{program:splc}\tag{F-SPLC}
\begin{split}
&  \mathrm{CP}_i(p,q,h) \leq  u_i(x-c_i) \quad \forall i \in A,\\
&   x_{ijk} \leq l_{ijk} \quad \forall i,j,k \in A\times G\times [l],\\
&   \sum_{i \in A}\sum_{k \in [l]} x_{ijk} \leq 1 \quad \forall j \in G,\\
&   \sum_{j \in G}\sum_{k \in [l]} x_{ijk} \leq 1 \quad \forall i \in A,\\
& \sum_{j\in G} p_j +  \sum_{i\in A} q_i + \sum_{i,j,k\in A \times G \times [l]} h_{ijk} =n + \sum_{i \in A} c_i \min_{j,k \in G\times[l]} \frac{p_j+q_i+h_{ijk}}{u_{ijk}},\\
& x,p,q,h\geq 0
\end{split}
\end{equation}
where \begin{maxi*}
    {}
    {\sum_{j,k \in G \times [l]} u_{ijk}y_{jk}}
    {}
    {\mathrm{CP}_i(p,q,h) \coloneqq}
    \addConstraint{\sum_{j,k \in G\times [l]} (p_j+q_i+h_{ijk})  y_{jk}}{\leq 1+c_i \min_{j,k \in G\times[l]} \frac{p_j+q_i+h_{ijk}}{u_{ijk}}}
    \addConstraint{y}{\geq 0.}
\end{maxi*}

Following the exact proof steps of Lemma \ref{lem:equivalence} (using KKT conditions for SPLC utilities), it can also be shown that feasibility of (\ref{program:splc}) is equivalent to solving \ref{eq.CP-SPLC}.
\begin{lemma}[Optimization to Feasibility (SPLC)]\label{lem:equivalence2} Let $x^*$ be an optimal solution of Program (\ref{eq.CP-SPLC}), then there exist $p^*,q^*,h^*$ so that $(x^*,p^*,q^*,h^*)$ is feasible for (\ref{program:splc}). Moreover assume that $(x^*,p^*,q^*,h^*)$ is a feasible solution for (\ref{program:splc}), then $x^*$ is an optimal solution for Program (\ref{eq.CP-SPLC}).
\end{lemma}

Moreover, Multiplicative Weights Update Algorithm (see Algorithm \ref{alg:MWUsplc}) is modified so that it accounts for the constraints on segments (i.e., $x_{ijk} \leq l_{ijk}$), i.e., we introduce extra variables $h$ that penalize the allocations $x_{ijk}$ that exceed the corresponding upper bound $l_{ijk}.$

\begin{algorithm}[h!]
  \caption{Multiplicative Weight Update for SPLC}


  \begin{algorithmic}

  \State $\tilde{p}^{(0)} \leftarrow\textbf{1}, \tilde{q}^{(0)} \leftarrow\textbf{1}, \tilde{h}^{(0)}\leftarrow\textbf{1}$ \/\/ initialization

  \For{$t=1$ to $T$}

     \State \textbf{Rescale} $\tilde{p}^{(t)}, \tilde{q}^{(t)}, \tilde{h}^{(t)}$ so that
     \State \;\;\;\;\;\;\;\;$\sum_{j \in G} p_j^{(t)}+ \sum_{i \in A} q_i^{(t)} + \sum_{i,j,k \in A\times G \times [l]} h_{ijk}^{(t)}= n + \sum_{i \in A} c_i \min_{j,k \in G\times [l]} \frac{p_j^{(t)}+q_i^{(t)}+h_{ijk}^{(t)}}{u_{ijk}}.$

     \For{$i=1$ to $n$}

     \State  $x_i^{(t)} \leftarrow \arg\max_y \left\{\sum_{j,k \in G\times[l]} u_{ijk}y_{jk} : \sum_{j \in G}\sum_{k \in[l]}\left(p^{(t)}_j+q_i^{(t)}+h_{ijk}^{(t)}\right) y_{jk} \leq 1+c_i \min_{j \in G} \frac{p_j^{(t)}+q_i^{(t)}+h_{ijk}^{(t)}}{u_{ijk}}\right\}.$
     \EndFor

     \State For all $j,\; d_{G,j}^{(t)} \leftarrow \sum_{i\in A}\sum_{k \in [l]} x_{ijk}^{(t)}$ and for all $i,\;d_{A,i}^{(t)} \leftarrow \sum_{j \in G}\sum_{k \in [l]} x_{ijk}^{(t)}$

     \State  $\sigma^{(t)} \leftarrow \min \left(\frac{1}{\max_{j\in G} d_{G,j}^{(t)}} , \frac{1}{\max_{i\in A} d_{A,i}^{(t)}}, \min_{i,j,k A\times G\times[l]} \frac{1}{x_{ijk}^{(t)}}\right)$

     \For{$j=1$ to $n$}
     \State $\tilde{p}_j^{(t+1)} \leftarrow \tilde{p}_j^{(t)} (1+\epsilon \sigma^{(t)}d_{G,j}^{(t)})$
     \EndFor

     \For{$i=1$ to $n$}
     \State $\tilde{q}_i^{(t+1)} \leftarrow \tilde{q}_i^{(t)} (1+\epsilon \sigma^{(t)}d_{A,i}^{(t)})$
     \EndFor
     \For{$i=1$ to $n$}
         \For{$i=1$ to $n$}
           \For{$k=1$ to $l$}
            \State $\tilde{h}_{ijk}^{(t+1)} \leftarrow \tilde{h}_{ijk}^{(t)} \left(1+\epsilon \sigma^{(t)}x^{(t)}_{ijk}\right)$
           \EndFor
          \EndFor
     \EndFor
  \EndFor

  \State \textbf{return } $\overline{x} \leftarrow \frac{\sum_{t=1}^T \sigma^{(t)}x^{(t)}}{\sum_{t=1}^T \sigma^{(t)}}$ (matrix), $\overline{p} \leftarrow \frac{\sum_{t=1}^T \sigma^{(t)}p^{(t)}}{\sum_{t=1}^T \sigma^{(t)}}$, $\overline{q} \leftarrow \frac{\sum_{t=1}^T \sigma^{(t)}q^{(t)}}{\sum_{t=1}^T \sigma^{(t)}}$, $\overline{h} \leftarrow \frac{\sum_{t=1}^T \sigma^{(t)}h^{(t)}}{\sum_{t=1}^T \sigma^{(t)}}$.
  \end{algorithmic}
  \label{alg:MWUsplc}
\end{algorithm}

Before we proceed with the analysis of Algorithm \ref{alg:MWUsplc}, we must note that the rescaling step (first step inside the outer For-loop) is well-defined and the proof of the claim is identical to the proof of Lemma \ref{lem:scale}.
\paragraph{Potential function}
To analyze the rate of convergence of Algorithm \ref{alg:MWUsplc}, we will use a potential function argument (in a similar way as we did for linear utilities). The potential function will be given by:
\begin{equation}\label{eq:potentialsplc}
\Phi(t) = \sum_{j \in G} \tilde{p}_j^{(t)} + \sum_{i \in A} \tilde{q}_i^{(t)}+\sum_{i,j,k \in A\times G\times[l]} \tilde{h}_{ijk}^{(t)}.
\end{equation}

We need to prove upper and lower bounds as before. The upper bound and the lower bound in this case can be summarized in the inequalities below (proof follows same steps as in \ref{lem:upper} and \ref{lem:lower} and is omitted):
\begin{equation}\label{eq:combinedsplc}
\epsilon(1-\epsilon) \max\left(\max_{j \in G} \sum_{t=1}^{T} \sigma^{(t)}d_{G,j}^{(t)},\max_{i \in A} \sum_{t=1}^{T} \sigma^{(t)}d_{A,i}^{(t)}, \max_{i,j,k}
\sum_{t=1}^T \sigma^{(t)}d_{ijk}^{(t)}\right)\leq \ln \Phi(T) \leq \ln (2n) + \epsilon \sum_{t=1}^{T-1}\sigma^{(t)}.
\end{equation}

\paragraph{Finishing the proof of Theorem \ref{thm:main} for SPLC utilities.}
Setting $\overline{x} = \frac{\sum_{t=1}^T \sigma^{(t)}x^{(t)}}{\sum_{t=1}^T \sigma^{(t)}}$ and $\overline{p} = \frac{\sum_{t=1}^T \sigma^{(t)}p^{(t)}}{\sum_{t=1}^T \sigma^{(t)}}$, $\overline{q} = \frac{\sum_{t=1}^T \sigma^{(t)}q^{(t)}}{\sum_{t=1}^T \sigma^{(t)}}$ and $\overline{h} = \frac{\sum_{t=1}^T \sigma^{(t)}h^{(t)}}{\sum_{t=1}^T \sigma^{(t)}}$ (average among all allocations, prices at all times, the output of MWUA \ref{alg:MWUsplc}). $\overline{x}, \overline{p}, \overline{q},\overline{h}$ are non-negative and satisfy the prices constraints (because of the rescaling step). Moreover functions $\textrm{CP}_i(p,q,h)$ are convex. As long as $T$ is chosen to be $\Theta\left(\frac{2n\log 2n}{\epsilon^2}\right)$, it holds
\begin{equation}\label{eq:needonesplc}
\mathrm{CP}_i(\overline{p},\overline{q},\overline{h}) \leq u_i (\overline{x}) \textrm{ for all }i \in A,
\end{equation}
and moreover
 \begin{equation}\label{eq:needtwosplc}
 \begin{split}
&\sum_{i\in A}\sum_{k\in[l]} \overline{x}_{i jk} \leq 1 + O(\epsilon) \textrm{ for all }j \in G,
\sum_{j\in G}\sum_{k\in[l]} \overline{x}_{i jk} \leq 1 + O(\epsilon) \textrm{ for all }i \in A\textrm{ and }\\&
\overline{x}_{ijk} \leq l_{ijk}+ O(\epsilon) \textrm{ for all }i,j,k \in A\times G\times[l].
\end{split}
\end{equation}

As long as we have (\ref{eq:needtwosplc}), we rescale appropriately so that $\sum_{ik} \overline{x}_{ijk} \leq 1$ for all $j \in G$ and $\sum_{jk} \overline{x}_{ijk} \leq 1$ for all $i \in A$. After rescaling, we push the allocations greedily for every agent $i$ so that if segment $1 \leq t \leq l $ is not fully occupied (i.e., $0<\overline{x}_{ijt} < l_{ijt}$) then $\overline{x}_{ijt-1} = l_{ijt-1}.$ it turns out that the resulting allocation $\overline{x}$ will be an $O(\epsilon)$-approximate solution for each agent in Program (\ref{eq.CP-SPLC}).

%% file: conditional_gradient.tex
\section{Conditional Gradient Algorithms}
\label{sec.cond_grad}

Since the problem of finding Nash-bargaining points can be captured by a convex program, we can leverage techniques from convex optimization such as gradient descent to compute rapidly converging approximations.
In this section we will provide a conditional gradient type algorithm which is relatively simple and projection free while converging at a rate of $O(1 / \epsilon)$.

\subsection{Linear Utilities (LiF)}

Let us first consider the simplest setting of a one-sided market with linear utilities, i.e.\ the model LiF.
We are given a set $A$ of agents and $G$ of goods with $|A| = |G| = n$ and utilities $u_{i j}$ for all $i \in A$ and $j \in G$.
Our goal is to solve the convex program
\begin{maxi}
    {}
    {\sum_{i \in A} \log u_i(x)}
    {}
    {\label{cp:mat_eg_linear}}
    \addConstraint{\sum_{i \in A} x_{i j}}{\leq 1}{\quad \forall j \in G}
    \addConstraint{\sum_{j \in G} x_{i j}}{\leq 1}{\quad \forall i \in A}
    \addConstraint{x}{\geq 0}
\end{maxi}
where $u_i(x) \coloneqq \sum_{j \in G} u_{i j} x_{i j}$.

In the following we will assume that the utilities have been rescaled so that $\max \{ u_{i j} \mid j \in G\} = 1$ for all $i$ and $\sum_{j \in G} u_{i j} \geq \frac{1}{\kappa}$ for some value $\kappa$.
In particular, the objective function is 1-strongly concave with respect to $u_i(x)$ but not $x_{i j}$.
Note that since $u_i(x)$ is not bounded from below, the objective function is neither Lipschitz nor smooth.

Recall also that we have shown in Lemma~\ref{lem:opt_lb} that
\[
    u_i(x) = \sum_{j \in G} u_{i j} x_{i j} \geq \frac{1}{2 n} \sum_{j \in G} u_{i j} \geq \frac{1}{2 \kappa n}.
\]
for any optimum solution $x$.
This implies that we could restrict ourselves to the problem
\begin{maxi*}
    {}
    {\sum_{i \in A} \log u_i(x)}
    {}
    {}
    \addConstraint{u_i(x)}{\geq \frac{1}{2 \kappa n}}{\quad \forall i \in A}
    \addConstraint{\sum_{i \in A} x_{i j}}{\leq 1}{\quad \forall j \in G}
    \addConstraint{\sum_{j \in G} x_{i j}}{\leq 1}{\quad \forall i \in A}
    \addConstraint{x}{\geq 0.}
\end{maxi*}
without changing the optimum solution.
However, while general purpose projected gradient descent algorithms can achieve $O(\log(1 / \epsilon))$ convergence rates on this modified problem, we wish to exploit the combinatorial structure of $P$.
Therefore we will modify the objective function instead.

Define the quadratic extension of the logarithm at $x_0 > 0$ by
\[
    \eta(x; x_0) \coloneqq \begin{cases}
        \log(x_0) + (x - x_0) \frac{1}{x_0} - \frac{1}{2} (x - x_0)^2 \frac{1}{x^2_0} & \text{if } x \leq x_0, \\
        \log(x) & \text{otherwise.}
    \end{cases}
\]
Then $\eta(x; x_0)$ is $\frac{1}{x^2_0}$-smooth everywhere, i.e.\ its gradient is $\frac{1}{x^2_0}$-Lipschitz.

Consider now the modified convex program
\begin{maxi}
    {}
    {\sum_{i \in A} \eta\left(u_i(x); \frac{1}{2 \kappa n}\right)}
    {}
    {\label{cp:mat_eg_mod}}
    \addConstraint{\sum_{i \in A} x_{i j}}{\leq 1}{\quad \forall j \in G}
    \addConstraint{\sum_{j \in G} x_{i j}}{\leq 1}{\quad \forall i \in A}
    \addConstraint{x}{\geq 0.}
\end{maxi}

\begin{lemma}
    Every optimum solution to \eqref{cp:mat_eg_mod} is also optimum for \eqref{cp:mat_eg_linear}.
\end{lemma}

\begin{proof}
    First let $x$ be the optimum solution for \eqref{cp:mat_eg_linear}.
    By Lemma~\ref{lem:opt_lb}, the objective function of the two programs agrees up to the first-order at $x$ and thus $x$ is also optimum for \eqref{cp:mat_eg_mod}.
    But now any other optimum solution $x'$ for the modified problem must satisfy $u_i(x') = u_i(x)$ for all $i$ since the objective is strictly concave in $u_i$.
    Thus it is also an optimum solution for \eqref{cp:mat_eg_linear}.
\end{proof}

In the following let
\begin{align*}
    \phi(x) &\coloneqq \sum_{i \in A} \log u_i(x), \\
    \psi(x) &\coloneqq \sum_{i \in A} \eta\left(u_i(x); \frac{1}{2 \kappa n}\right)
\end{align*}
be the objective functions of \eqref{cp:mat_eg_linear} and \eqref{cp:mat_eg_mod} respectively.

\begin{lemma}\label{lem:orig_conv}
    Let $x$ be an arbitrary allocation and $x^*$ an optimal one.
    Let $\psi(x^*) - \psi(x) \leq \epsilon$ for some $\epsilon > 0$.
    Then $||u(x) - u(x^*)||^2 \leq 2 \epsilon$.
    Moreover, if $\delta \in (0, 1/2)$, and $\epsilon = \min\left\{\delta, \frac{\delta^{2/3}}{8 n^2 \kappa^2}\right\}$, then $\phi(x^*) - \phi(x) \leq O(1) \delta$.
\end{lemma}

\begin{proof}
    The first part follows directly from the fact that $\psi$ is 1-strongly concave over the feasible region.
    Now observe that by Taylor's theorem we have
    \[
        \eta\left(u_i(x); \frac{1}{2 n \kappa}\right) - \log u_i(x) \leq \frac{2}{3} \frac{\left|u_i(x) - \frac{1}{2 n \kappa}\right|^3}{u_i(x)^3}.
    \]
    But because $u(x^*)_i \geq \frac{1}{2 n \kappa}$ and $||u(x) - u(x^*)||^2 \leq 2 \epsilon$, we know that $\left|u_i(x) - \frac{1}{2 n \kappa}\right| \leq |u_i(x) - u_i(x^*)|$ and
    \[
        u_i(x) \geq \frac{1}{2 n \kappa} - \sqrt{2 \epsilon} \geq (1 - \delta^{1/3}) \frac{1}{2 n \kappa}.
    \]
    
    Finally, we compute
    \begin{align*}
        \phi(x^*) - \phi(x) &= \psi(x^*) - \psi(x) + \psi(x) - \phi(x) \\
                            &\leq \delta + \sum_{i \in A} \left(\eta\left(u_i(x); \frac{1}{2 n \kappa}\right) - \log u_i(x)\right) \\
                            &\leq \delta + O(1) n^3 \kappa^3 \sum_{i \in A} |u_i(x) - u_i(x^*)|^3 \\
                            &\leq \delta + O(1) n^3 \kappa^3 ||u(x) - u(x^*)||^{3 / 2} \\
                            &\leq O(1) \delta. \qedhere
    \end{align*}
\end{proof}

These two lemmas together imply that it suffices to solve \eqref{cp:mat_eg_mod}.
Approximate solutions to this modified program are also approximate solutions for \eqref{cp:mat_eg_linear}, both wrt.\ the Euclidean norm on the utility vectors and the original objective function.
Finding an approximate solution in a combinatorial way can be achieved using the conditional gradient method over the matching polytope; see Algorithm~\ref{alg:fw_eg_mod}.
Note that the gradient of $\psi$ is easily computable.

\begin{algorithm}
    \begin{algorithmic}[1]
        \State $x^{(0)} \equiv 0$
        \For{$t \leftarrow 1, \ldots, T$}
            \For{$i \in A, j \in G$}
                \State $w^{(t)}_{i j} \leftarrow \partial_{i j} \psi(x^{(t - 1)})$
            \EndFor
            \State $y^{(t)} \leftarrow \text{max-weight matching with weights } w^{(t)}$
            \State $x^{(t)} \leftarrow \left(1 - \frac{2}{t + 1}\right) x^{(t - 1)} + \frac{2}{t + 1} y^{(t)}$
        \EndFor
        \State \Return $x^{(T)}$
    \end{algorithmic}
    \caption{Conditional Gradient for program \eqref{cp:mat_eg_mod}\label{alg:fw_eg_mod}}
\end{algorithm}

\begin{theorem}\label{thm:fw_eg_mod_conv}
    Algorithm~\ref{alg:fw_eg_mod} returns some $x$ with $\psi(x^*) - \psi(x) \leq \epsilon$ in $O\left(\frac{n^3 \kappa^2}{\epsilon}\right)$ many iterations.
    Each iteration can be implemented in $O(n^3)$ time.
\end{theorem}

\begin{proof}
    The conditional gradient algorithm converges in $O\left(\frac{D^2 L}{\epsilon}\right)$ iterations where $D$ is the diameter of the polytope that is being optimized over and $L$ is the smoothness of the objective function.
    For a modern proof of this fact, see for example \cite{jaggi13revisiting}.
    In this case, the diameter of the matching polytope is $2 \sqrt{n}$ and $L = 4 n^2 \kappa^2$ since $\eta\left(u_i; \frac{1}{2 n \kappa}\right)$ is clearly $4 n^2 \kappa^2$-smooth.
    The amount of work in each iteration is $O(n^2)$ except for the computation of the max weight matching which can be done in $O(n^3)$ using the Hungarian method.
\end{proof}

We remark that it is in principle possible to achieve faster convergence rates for conditional gradient type algorithms by leveraging strong concavity in addition to smoothness \cite{garber16linear}.
Note that while $\psi$ is strongly concave in the utilities, it is unfortunately not strongly concave in the allocation $x$.

Nevertheless, $\psi$ can be written as $g(U x)$ where $g$ is strongly concave and $U x$ is the linear transformation of allocations to utilities.
There are more complex variants of conditional gradient methods, for example those which involve taking ``away steps'' \cite{jaggi2015global}, which can be shown to converge in $O(\log(1 / \epsilon))$ phases even in this more general setting.
However, these algorithms depend on difficult to compute Hoffman-type constants of $U$ relative to the matching polytope for which there is no known polynomial bound in the instance parameters.

\subsection{Non-Bipartite Matching Markets (NBLF)}

The result from the previous section can be extended to the non-bipartite setting, i.e.\ model NBLF.
Recall that we have a set of $n$ agents $A$ with non-negative utilities $u_{i j}$ for all $i, j \in A$.
The goal is then to solve the convex program
\begin{maxi}
    {}
    {\sum_{i \in A} \log u_i(x)}
    {}
    {\label{cp:mat_eg_nb}}
    \addConstraint{x(\delta(i))}{\leq 1}{\quad \forall i \in A}
    \addConstraint{x(E(B))}{\leq \frac{|B| - 1}{2}}{\quad \forall B \in \mathcal{O}}
    \addConstraint{x}{\geq 0}
\end{maxi}
where $u_i(x) = \sum_{j \in A} u_{i j} x_{i j}$ and $\mathcal{O}$ is the collection of all odd cardinality subsets of $A$.

Assume that the utilities have been rescaled so that $\max \{ u_{i j} \mid j \in A\} = 1$ for all $i$ and $\sum_{j \in A} u_{i j} \geq \frac{1}{\kappa}$ for some value $\kappa$.
Since it is possible to optimize over the matching polytope using combinatorial methods, the primary ingredient is to once again bound the utilities away from 0.

\begin{lemma}\label{lem:opt_lb_nb}
    Let $x$ be an optimal solution to \eqref{cp:mat_eg_nb}, then for all agents $i$ we have
    \[
        u_i(x) \geq \frac{1}{2 n^2} \sum_{j \in G} u_{i j} \geq \frac{1}{2 \kappa n^2}.
    \]
\end{lemma}

\begin{proof}
    The proof will be similar to the proof of Lemma~\ref{lem:opt_lb}, however the more complicated KKT conditions will yield a weaker bound.
    So let $p$ and $z$ be the optimal dual variables for our solution of \eqref{cp:mat_eg_nb}.

    Using the KKT conditions, we know that
    \[
        \frac{u_{i j}}{u_i(x)} + \frac{u_{j i}}{u_j(x)} \leq p_i + p_j + \sum_{\{i, j\} \subseteq B \in \mathcal{O}} z_B
    \]
    with equality when $x_{i j} > 0$.
    Therefore
    \begin{align*}
        2 \sum_{i \in A} p_i + \sum_{B \in \mathcal{O}} (|B| - 1) z_B &= \sum_{i \in A} \sum_{j \in A} \left(p_i + p_j + \sum_{\{i, j\} \subseteq B \in \mathcal{O}} z_B\right) x_{i j} \\
                                                                      &= \sum_{i \in A} \sum_{j \in A} \left(\frac{u_{i j} x_{i j}}{u_i(x)} + \frac{u_{j i} x_{i j}}{u_j(x)}\right) \\
                                                                      &= 2n.
    \end{align*}

    In addition, we can observe from the KKT conditions that
    \begin{align*}
        u_i(x) &= \max \left\{ \frac{u_{i j}}{p_i + p_j + \sum_{\{i, j\} \subseteq B \in \mathcal{O}} z_B
- u_{j i} / u_j(x)} \;\middle|\; j \in A\right\} \\
               &\geq \max \left\{ \frac{u_{i j}}{p_i + p_j + \sum_{\{i, j\} \subseteq B \in \mathcal{O}} z_B
} \;\middle|\; j \in A\right\} \\
               &\geq \frac{\sum_{j \in A} u_{i j}}{n p_i + \sum_{j \in A} p_j + \sum_{i \in B \in \mathcal{O}} (|B| - 1) z_B} \\
               &\geq \frac{1}{2 n^2} \sum_{j \in A} u_{i j}. \qedhere
    \end{align*}
\end{proof}

This bound may initially seem weak when compared to the $\frac{1}{2 n} \sum_{j \in G} u_{i j}$ bound from Lemma~\ref{lem:opt_lb} of the previous section.
However, one can show that this bound is tight up to a constant factor.

More specifically, consider for some parameter $\ell \in \mathbb{N}$, an instance that has $2 \ell + 1$ agents.
Agents $1, \ldots, \ell$ have utility 1 for agent $2 \ell + 1$ and 0 everywhere else.
Agents $\ell + 1, \ldots, 2 \ell$ are indifferent, i.e\ have utility 1 for all agents.
Finally, agent $2 \ell + 1$ has utility 1 for agents $\ell + 1, \ldots, 2 \ell$.
One can then show that in the optimal allocation, agent $2 \ell + 1$ will only get $\frac{1}{\ell + 1}$ units from their utility 1 edges.
Therefore
\[
    \frac{\sum_{j \in A} u_{i j}}{u_i(x^*)} = \ell (\ell + 1) \approx \frac{n^2}{4}.
\]

However, this utility bound is still adequate in order to show convergence of the conditional gradient method with the modified objective function
\[
    \psi(x) \coloneqq \sum_{i \in A} \eta\left(u_i(x); \frac{1}{2 n^2 \kappa}\right).
\]

\begin{theorem}
    Algorithm~\ref{alg:fw_eg_mod} returns some $x$ with $\psi(x^*) - \psi(x) \leq \epsilon$ in $O\left(\frac{n^5 \kappa^2}{\epsilon}\right)$ many iterations.
    Each iteration can be implemented in $O(n^3)$ time.
\end{theorem}

\begin{proof}
    The difference in the number of iterations compared to Theorem~\ref{thm:fw_eg_mod_conv} comes from the fact that $\psi$ is now $4 n^4 \kappa^2$-smooth.
    Each iteration can still be implemented in $O(n^3)$ time by using a weighted matching algorithm, now for non-bipartite graphs.
\end{proof}

\subsection{SPLC Utilities (SF)}

The approach from the previous sections can be readily extended to SPLC utilities, i.e.\ model SF.
In this case each agent $i$ has utilities $u_{i j k}$ for all $j \in G$ and $k \in \{1, \ldots, \ell\}$.
Note that for ease of notation we assume here that all utility functions have the same number of segments.
Moreover, each ``segment'' $i, j, k$ comes with a limit $l_{i j k}$ such that $\sum_{k = 1}^m l_{i j k} = 1$ for all agents $i$ and goods $j$.\footnote{We previously assumed that there is also a last segment that has unbounded length. However, since we will stay feasible at every step, this does not matter in this section.}
Recall that our goal is to solve the convex program
\begin{maxi!}
    {}
    {\sum_{i \in A} \log u_i(x)}
    {}
    {\label{cp:mat_eg_splc}}
    \addConstraint{\sum_{i \in A} \sum_{k = 1}^m x_{i j k}}{\leq 1}{\quad \forall j \in G}
    \addConstraint{\sum_{j \in G} \sum_{k = 1}^m x_{i j k}}{\leq 1}{\quad \forall i \in A}
    \addConstraint{x_{i j k}}{\leq l_{i j k}}{\quad \forall i \in A, j \in G, k \in [m]\label{cp:mat_eg_splc:c4}}
    \addConstraint{x}{\geq 0}
\end{maxi!}
where $u_i(x) = \sum_{j \in G} \sum_{k = 1}^\ell u_{i j k} x_{i j k}$.

Again we will assume that each agents utilities have been rescaled such that $\max \{u_{i j k} \mid j \in G, k \in [\ell]\} = 1$.
Moreover, assume that $\sum_{j \in G} \sum_{k = 1}^m u_{i j k} l_{i j k} \geq \frac{1}{\kappa}$ for some value $\kappa > 0$.
Recall that we showed
\[
    u_i(x) \geq \frac{1}{2 n} \sum_{j \in G} \sum_{k = 1}^\ell u_{i j k} l_{i j k} \geq \frac{1}{2 \kappa n}.
\]
for any optimum solution $x$ in Lemma~\ref{lem:opt_lb_splc}.
From here, the modified convex program is given by
\begin{maxi}
    {}
    {\sum_{i \in A} \eta\left(u_i(x); \frac{1}{2 \kappa n}\right)}
    {}
    {\label{cp:mat_eg_splc_mod}}
    \addConstraint{\sum_{i \in A} \sum_{k = 1}^\ell x_{i j k}}{\leq 1}{\quad \forall j \in G}
    \addConstraint{\sum_{j \in G} \sum_{k = 1}^\ell x_{i j k}}{\leq 1}{\quad \forall i \in A}
    \addConstraint{x_{i j k}}{\leq l_{i j k}}{\quad \forall i \in A, j \in G, k \in [\ell]}
    \addConstraint{x}{\geq 0}
\end{maxi}
and we may use the conditional gradient method to get an $\epsilon$-approximate solution; see Algorithm~\ref{alg:fw_eg_splc_mod}.
Here we make use of the operation $\mathrm{shift}(x)$ which takes a feasible solution $x$ to \eqref{cp:mat_eg_splc_mod} and for each agent $i$ and good $j$ it makes sure that the allocation uses segments with the largest value of $u_{i j k}$ first.
Assume wlog.\ that $u_{i j 1} > \dots > u_{i j \ell}$, then $\mathrm{shift}(x)$ can be recursively defined via
\[
    \mathrm{shift}(x)_{i j k} \coloneqq \min \left\{l_{i j k}, \left(\sum_{k' = 1}^\ell x_{i j k'} - \sum_{k' = 1}^{k - 1} \mathrm{shift}(x)_{i j k'}\right)_+\right\}.
\]
Moreover, we remark that the gradient is given by
\[
    \partial_{i j k} \psi(x) = \begin{cases}
        \frac{u_{i j k}}{u_i(x)} & \text{if } u_i \geq \frac{1}{2 n \kappa} \\
        4 n^2 \kappa^2 u_i(x) u_{i j k} & \text{otherwise}.
    \end{cases}
\]

\begin{algorithm}
    \begin{algorithmic}[1]
        \State $x^{(0)} \equiv 0$
        \For{$t \leftarrow 1, \ldots, T$}
            \For{$i \in A, j \in G, k \in [\ell]$}
                \State $w^{(t)}_{i j k} \leftarrow \partial_{i j k} \psi(x^{(t - 1)})$
            \EndFor
            \State $y^{(t)} \leftarrow \text{max-weight assignment with weights } w^{(t)} \text{ and edge capacities } l$
            \State $x^{(t)} \leftarrow \mathrm{shift}\left(\left(1 - \frac{2}{t + 1}\right) x^{(t - 1)} + \frac{2}{t + 1} \mathrm{shift}(y^{(t)})\right)$
        \EndFor
        \State \Return $x^{(T)}$
    \end{algorithmic}
    \caption{Conditional gradient for program \eqref{cp:mat_eg_splc_mod}\label{alg:fw_eg_splc_mod}}
\end{algorithm}

\begin{theorem}\label{thm:splc_main}
    Algorithm~\ref{alg:fw_eg_splc_mod} returns some $x$ with $\psi(x^*) - \psi(x) \leq \epsilon$ in $O\left(\frac{n^3 \kappa^2}{\epsilon}\right)$ many iterations.
    Each iteration can be implemented in $O(n^4 \ell^2 \log n)$ time.
\end{theorem}

\begin{proof}
    Recall that we generally have convergence in $O\left(\frac{D^2 L}{\epsilon}\right)$ iterations where $D$ is the diameter of the polytope that is being optimized over and $L$ is the smoothness of the objective function.
    Normally, i.e.\ without the shifting, the key ingredient the proof is observing that
    \[
        \psi(x^{(t)}) - \psi(x^{(t - 1)}) \geq \alpha_t \nabla \phi(x^{(t - 1)}) \cdot (y^{(t)} - x^{(t - 1)}) - \frac{L}{2} \alpha_t^2 ||y^{(t)} - x^{(t - 1)}||^2
    \]
    where $\alpha_t = \frac{2}{t + 1}$.
    Then one bounds
    \[
        \psi(x^{(t - 1)}) \cdot (y^{(t)} - x^{(t - 1)}) \geq \psi(x^{(t - 1)}) \cdot (x^* - x^{(t - 1)}) \geq \psi(x^*) - \psi(x^{(t - 1)})
    \]
    and $||y^{(t)} - x^{(t - 1)}|| \leq D$ from which one can then deduce the stated rate of convergence.
    
    Now note that $\psi(\mathrm{shift}(x)) \geq \psi(x)$ and $\nabla \psi(x) \cdot \mathrm{shift}(y) \geq \nabla \psi(x) \cdot y$.
    Thus we have
    \begin{align*}
        \psi(x^{(t)}) - \psi(x^{(t - 1)}) &\geq \psi((1 - \alpha_t) x^{(t - 1)} + \alpha_t \mathrm{shift}(y^{(t)})) - \psi(x^{(t - 1)}) \\
            &\geq \alpha_t \nabla \phi(x^{(t - 1)}) \cdot (\mathrm{shift}(y^{(t)}) - x^{(t - 1)}) \\
            &\phantom{\geq {}} - \frac{L}{2} \alpha_t^2 ||\mathrm{shift}(y^{(t)}) - x^{(t - 1)}||^2 \\
            &\geq \alpha_t \nabla \phi(x^{(t - 1)}) \cdot (y^{(t)} - x^{(t - 1)}) \\
            &\phantom{\geq {}} - \frac{L}{2} \alpha_t^2 ||\mathrm{shift}(y^{(t)}) - x^{(t - 1)}||^2.
    \end{align*}
    
    The advantage of this approach is that we now need to bound $||\mathrm{shift}(y^{(t)}) - x^{(t - 1)}||^2$ instead of the distance between two potentially arbitrary points in the polytope.
    Indeed by the subadditivity of the square function for positive values, it now follows that
    \[
        ||\mathrm{shift}(y^{(t)}) - x^{(t - 1)}||^2 \leq \sum_{i \in A} \sum_{j \in G} \left(\sum_{k = 1}^\ell y^{(t)}_{i j k} - \sum_{k = 1}^\ell x^{(t - 1)}_{i j k}\right)^2.
    \]
    This allows us to conclude $||\mathrm{shift}(y^{(t)}) - x^{(t - 1)}||^2 \leq 2 \sqrt{n}$ because the vectors $\left(\sum_{k = 1}^\ell y^{(t)}_{i j k}\right)_{i j}$ and $\left(\sum_{k = 1}^\ell x^{(t - 1)}_{i j k}\right)_{i j}$ are part of the bipartite matching polytope on agents $A$ and goods $G$.
    
    As in the case of linear utilities, the objective function is $4 n^2 \kappa^2$-smooth and so we get convergence in $O\left(\frac{n^3 \kappa^2}{\epsilon}\right)$.
    The time per iteration is clearly bottlenecked by the computation of the max-weight assignment with edge capacities.
    This problem is a special case of computing a min-cost flow with $2 n$ vertices and $n \ell$ capacitated edges which may be solved in $O(n^4 \ell^2 \log(n))$ time using Orlin's algorithm via a classic reduction from the transshipment problem \cite{orlin93faster} or Vygen's algorithm \cite{vygen2002dual}.
\end{proof}

We remark that in practice, computing min-cost flows is far more efficient than the stated worst case running time.
This is especially the case when one has a ``warm start'', i.e.\ a solution that is already close to optimal.
In the case of algorithm~\ref{alg:fw_eg_splc_mod}, one may use $y^{(t - 1)}$ as a starting solution for the min-cost flow algorithm that computes $y^{(t)}$.

\subsection{Extension to Endowments (LiAD)}

Let us now return to the case of linear utilities but in which agents come with preexisting endowments, i.e.\ each agent $i$ has some disagreement utility $c_i$ and the goal is then to solve the convex program
\begin{maxi}
    {}
    {\sum_{i \in A} \log (u_i(x) - c_i)}
    {}
    {\label{cp:mat_eg_linear_ci}}
    \addConstraint{\sum_{i \in A} x_{i j}}{\leq 1}{\quad \forall j \in G}
    \addConstraint{\sum_{j \in G} x_{i j}}{\leq 1}{\quad \forall i \in A}
    \addConstraint{x}{\geq 0.}
\end{maxi}
where $u_i(x) = \sum_{j \in G} u_{i j} x_{i j}$.
This is model LiAD.

The added difficulty of this setting is that, in general, a feasible solution may not exist.
We thus assume that we are given a feasible solution $(\hat{x}, \delta)$ of the LP
\begin{maxi}
    {}
    {\delta}
    {}
    {\label{lp:mat_eg_linear_ci}}
    \addConstraint{u_i(x)}{\geq (1 + \delta) c_i}{\quad \forall i \in A}
    \addConstraint{\sum_{i \in A} x_{i j}}{\leq 1}{\quad \forall j \in G}
    \addConstraint{\sum_{j \in G} x_{i j}}{\leq 1}{\quad \forall i \in A}
    \addConstraint{x}{\geq 0}
    \addConstraint{\delta}{\geq 0.}
\end{maxi}
with $\delta > 0$.
We call this $\delta$ the \emph{feasibility gap} of the instance.

In practice there are two likely scenarios as to how one might obtain $\delta$.
Either one solves the above linear program to find the optimal $\delta$ or the disagreement utilities $c_i$ are defined as some constant fraction of the agents' utilities over their initial endowments.
More precisely, assume that each agent $i$ comes to the market with endowments $e_{i j}$ over the goods $j$.
Then one may simply define $c_i \coloneqq \frac{1}{1 + \delta} \sum_{j \in G} u_{i j} e_{i j}$.
The market will then guarantee that no agent gets worse by a factor of more than $(1 + \delta)$ which is an approximate notion of individual rationality.

Regardless of how one obtains a feasible solution to \eqref{lp:mat_eg_linear_ci}, we may use it to derive a similar lower bound as in Lemma~\ref{lem:opt_lb}.
Once again, assume that the utilities have been rescaled so that $\max \{ u_{i j} \mid j \in G\} = 1$ for all $i$ and $\sum_{j \in G} u_{i j} \geq \frac{1}{\kappa}$ for some $\kappa > 0$.

\begin{lemma}\label{lem:opt_lb_linear_ci}
    Let $x$ be an optimal solution to \eqref{cp:mat_eg_linear_ci}, then for all agents $i$ we have
    \[
        u_i(x) \geq c_i + \frac{1}{2 n^2 (1 + 1 / \delta)} \sum_{j \in G} u_{i j} \geq c_i + \frac{1}{2 n^2 (1 + 1 / \delta) \kappa}.
    \]
\end{lemma}

\begin{proof}
    Let $p_j$ and $q_i$ as always be optimum dual variables for the matching constraints in \eqref{cp:mat_eg_linear_ci}.
    Then by the KKT conditions we have that $\frac{u_{i j}}{u_i(x) - c_i} \leq p_j + q_i$ with equality if $x_{i j} > 0$.

    From this complementarity we can deduce that
    \begin{align*}
        u_i(x) - c_i &= \max\left\{\frac{u_{i j}}{p_j + q_i} \;\middle|\; j \in G\right\} \\
    \shortintertext{and}
        \frac{u_i(x)}{u_i(x) - c_i} &= \sum_{j \in G} {x_{i j} (p_j + q_i)}.
    \end{align*}
    Together this implies that
    \[
        \sum_{j \in G} {x_{i j} (p_j + q_i)} = 1 + c_i \min\left\{\frac{p_j + q_i}{u_{i j}} \;\middle|\; j \in G\right\}.
    \]
    
    Now we may use our feasible solution $(\hat{x}, \delta)$ to \eqref{lp:mat_eg_linear_ci} in order to bound
    \begin{align*}
        \sum_{j \in G} p_j + \sum_{i \in A} q_i &= \sum_{i \in A} \sum_{j \in G} x_{i j} (p_j + q_i) \\
            &= n + \sum_{i \in A} c_i \min\left\{\frac{p_j + q_i}{u_{i j}} \;\middle|\; j \in G\right\} \\
            &\leq n + \frac{1}{1 + \delta} \sum_{i \in A} \sum_{j \in G} u_{i j} \hat{x}_{i j} \min\left\{\frac{p_{j'} + q_i}{u_{i j'}} \;\middle|\; j' \in G\right\} \\
            &\leq n + \frac{1}{1 + \delta} \sum_{i \in A} \sum_{j \in G} \hat{x}_{i j} (p_{j} + q_i) \\
            &\leq n + \frac{1}{1 + \delta} \left(\sum_{j \in G} p_j + \sum_{i \in A} q_i\right)
    \end{align*}
    which implies in particular that $\sum_{j \in G} p_j + \sum_{i \in A} q_i \leq \left(1 + \frac{1}{\delta}\right) n$.
    
    Finally, using the same idea as in Lemma~\ref{lem:opt_lb}, we get
    \begin{align*}
        u_i(x) &= c_i + \max\left\{\frac{u_{i j}}{p_j + q_i} \;\middle|\; j \in G\right\} \\
            &\geq c_i + \frac{\sum_{j \in G} u_{i j}}{n q_i + \sum_{j \in G} p_j} \\
            &\geq c_i + \frac{1}{2 n^2 (1 + 1 / \delta)} \sum_{j \in G} u_{i j}. \qedhere
    \end{align*}
\end{proof}

The modified convex program will thus use the objective function
\[
    \psi(x) \coloneqq \sum_{i \in A} \eta\left(u_i(x) - c_i; \frac{1}{2 n^2 (1 + 1 / \delta) \kappa}\right)
\]
and we can now use Algorithm~\ref{alg:fw_eg_mod} to solve the modified convex program with the new gradients
\[
    \partial_{i j} \psi(x) = \begin{cases}
        \frac{u_{i j}}{u_i(x)} & \text{if } u_i(x) \geq \frac{1}{2 n^2 (1 + 1/\delta) \kappa} \\
        4 n^4 (1 + 1 / \delta)^2 \kappa^2 u_i(x) u_{i j} & \text{otherwise}.
    \end{cases}
\]

\begin{theorem}
    Algorithm~\ref{alg:fw_eg_mod} returns $x$ with $\psi(x^*) - \psi(x) \leq \epsilon$ in $O\left(\frac{n^5 (1 + 1/\delta)^2 \kappa^2}{\epsilon}\right)$ many iterations.
    Each iteration can be implemented in $O(n^3)$ time.
\end{theorem}

\begin{proof}
    The only difference to Theorem~\ref{thm:fw_eg_mod_conv} is that the objective function $\psi$ is now $4 n^4 (1 + 1/\delta)^2 \kappa^2$-smooth instead of $4 n^2 \kappa^2$ as before.
\end{proof}

We remark that an analogous result to Lemma~\ref{lem:orig_conv} holds also in this setting, i.e.\ convergence in the modified objective $\psi$ implies convergence of the agents' utilities in the Euclidean norm.
Thus, if one chooses $\epsilon$ on the order of $O\left(\frac{1}{n \sqrt{(1 + 1/\delta) \kappa}}\right)$, then one is guaranteed that $u_i(x) \geq c_i$ for all agents $i$.

\subsection{Two-Sided SPLC Utilities with Endowments (2SAD)}\label{sec.2SAD}

Let us finally consider the most complex model which includes two-sided SPLC utilities and endowments, i.e.\ model 2SAD.
There two sides of agents $A$ and jobs $J$ with $|A| = |J| = n$.
Each agent $i \in A$ has SPLC utilities over agents $j \in J$, defined as before by segments with utility $u_{i j k}$ and allocation limit $l_{i j k}$ for $k \in [\ell]$.\footnote{We once again assume for the sake of simplicity that all utility functions have an equal number of segments.}
Likewise, each agent $j \in J$ has SPLC utilities over the agents $i \in A$ and these are defined by segments with utilities $w_{i j k}$ and the same limits $l_{i j k}$ for $k \in [\ell]$.
These segments should satisfy $u_{i j 1} \geq \dots \geq u_{i j \ell}$, $w_{i j 1} \geq \dots \geq w_{i j \ell}$ and $\sum_{k = 1}^\ell l_{i j k} = 1$.
Finally, we assume that we are given disagreement utilities $c_i$ for agents in $A$ and $d_j$ for agents in $J$.

The goal is then to solve the convex program
\begin{maxi!}
    {}
    {\sum_{i \in A} \log (u_i(x) - c_i) + \sum_{j \in J} \log (w_j(x) - d_j)}
    {}
    {\label{cp:mat_eg_ts}}
    \addConstraint{\sum_{i \in A} \sum_{k = 1}^\ell x_{i j k}}{\leq 1}{\quad \forall j \in J}
    \addConstraint{\sum_{j \in J} \sum_{k = 1}^\ell x_{i j k}}{\leq 1}{\quad \forall i \in A}
    \addConstraint{x_{i j k}}{\leq l_{i j k}}{\quad \forall i \in A, j \in J, k \in [\ell]\label{cp:mat_eg_ts:c4}}
    \addConstraint{x}{\geq 0}
\end{maxi!}
where $u_i(x) \coloneqq \sum_{j \in G} \sum_{k = 1}^\ell u_{i j k} x_{i j k}$ and likewise for $w_j$.

We assume that the utilities have been scaled so that for all $i \in A$, we have $\max\{u_{i j k} \mid j \in J, k \in [m]\} = 1$ and $\sum_{j \in J} \sum_{k = 1}^m l_{i j k} u_{i j k} \geq \frac{1}{\kappa}$ for some $\kappa$.
Symmetrically, we require $\max\{w_{i j k} \mid i \in A, k \in [m]\} = 1$ and $\sum_{i \in A} \sum_{k = 1}^m l_{i j k} w_{i j k} \geq \frac{1}{\kappa}$ for all $j \in J$.
Finally, we also assume that there exists a feasible solution $\hat{x}$ with $u_i(\hat{x}) \geq c_i (1 + \delta)$ and $w_j(\hat{x}) \geq d_j (1 + \delta)$ for all $i$ and $j$.

\begin{lemma}\label{lem:opt_lb_all}
    Let $x$ be an optimal solution to \eqref{cp:mat_eg_ts}, then for all agents $i$ and $j$ we have
    \begin{align*}
        u_i(x) &\geq c_i + \frac{1}{2 n^2 (1 + 1 / \delta) \kappa}, \\
        w_j(x) &\geq d_j + \frac{1}{2 n^2 (1 + 1 / \delta) \kappa}.
    \end{align*}
\end{lemma}

\begin{proof}
    We essentially need to combine the proofs of Lemma~\ref{lem:opt_lb_splc} and Lemma~\ref{lem:opt_lb_linear_ci}.
    So let $p_j$ and $q_i$ be optimum dual variables for the matching constraints and let $h_{i j k}$ be the optimum dual variables for the constraints \eqref{cp:mat_eg_splc:c4}.
    The KKT conditions yield that $\frac{u_{i j k }}{u_i(x) - c_i} + \frac{w_{i j k}}{w_j(x) - d_j} \leq p_j + q_i + h_{i j k}$ with equality when $x_{i j k} > 0$.

    In particular, we have
    \[
        u_i(x) - c_i = \max \left\{\frac{u_{i j k}}{p_j + q_i + h_{i j k} - w_{i j k} / (w_j(x) - d_j)} \;\middle|\; j \in J, k \in [\ell]\right\}
    \]
    and thus
    \[
        \frac{u_i(x)}{u_i(x) - c_i} = 1 + c_i \min \left\{\frac{p_j + q_i + h_{i j k} - w_{i j k} / (w_j(x) - d_j)}{u_{i j k}} \;\middle|\; j \in J, k \in [\ell]\right\}.
    \]
    
    Making use of our feasible allocation $\hat{x}$ which satisfies $u_i(\hat{x}) \geq (1 + \delta) c_i$ and $w_j(\hat{x}) \geq (1 + \delta) d_j$ we can then compute
    \begin{align*}
        &\phantom{={}} \sum_{i \in A} \sum_{j \in J} \sum_{k = 1}^\ell l_{i j k} h_{i j k} + \sum_{j \in J} p_j + \sum_{i \in A} q_i \\
        &= \sum_{i \in A} \frac{u_i}{u_i - c_i} + \sum_{j \in J} \frac{w_j}{w_j - d_j} \\
        &= n + \sum_{i \in A} c_i \min \left\{\frac{p_j + q_i + h_{i j k} - w_{i j k} / (w_j - d_j)}{u_{i j k}} \;\middle|\; j \in J, k \in [\ell]\right\} \\
        &\phantom{={}} + \sum_{j \in J} \frac{w_j}{w_j - d_j} \\
        &\leq n + \frac{1}{1 + \delta} \left(\sum_{i \in A} \sum_{j \in J} \sum_{k = 1}^\ell l_{i j k} h_{i j k} + \sum_{j \in J} p_j + \sum_{i \in A} q_i - \sum_{j \in J} \frac{\hat{w}_j}{w_j - d_j}\right) \\
        &\phantom{={}} + \sum_{j \in J} \frac{w_j}{w_j - d_j} \\
        &\leq n + \frac{1}{1 + \delta} \left(\sum_{i \in A} \sum_{j \in J} \sum_{k = 1}^\ell l_{i j k} h_{i j k} + \sum_{j \in J} p_j + \sum_{i \in A} q_i\right) \\
        &\phantom{={}} - \sum_{j \in J} \frac{d_j}{w_j - d_j} + \sum_{j \in J} \frac{w_j}{w_j - d_j} \\
        &= 2 n + \frac{1}{1 + \delta} \left(\sum_{i \in A} \sum_{j \in J} \sum_{k = 1}^\ell l_{i j k} h_{i j k} + \sum_{j \in J} p_j + \sum_{i \in A} q_i\right).
    \end{align*}
    
    This implies that
    \[
        \sum_{i \in A} \sum_{j \in J} \sum_{k = 1}^m l_{i j k} h_{i j k} + \sum_{j \in J} p_j + \sum_{i \in A} q_i \leq 2 n \left(1 + \frac{1}{\delta}\right)
    \]
    and thus
    \begin{align*}
        u_i &= c_i + \max \left\{\frac{u_{i j k}}{p_j + q_i + h_{i j k } - w_{i j k} / (w_j - d_j)} \;\middle|\; j \in J, k \in [m]\right\} \\
            &\geq c_i + \max \left\{\frac{u_{i j k}}{p_j + q_i + h_{i j k }} \;\middle|\; j \in J, k \in [m]\right\} \\
            &\geq \sum_{j \in J} \sum_{k = 1}^m \frac{u_{i j k}}{p_j + q_i + h_{i j k}} \cdot \frac{l_{i j k} \cdot (p_j + q_i + h_{i j k})}{\sum_{j' \in J}p_{j'} + n q_i + \sum_{j' \in J} \sum_{k' = 1}^m l_{i j' k'} h_{i j' k'}} \\
            &\geq \frac{\sum_{j \in J} \sum_{k = 1}^m l_{i j k} u_{i j k}}{2 n^2 (1 + 1 / \delta)}.
    \end{align*}
    By symmetry the same bound holds for all $w_j$.
\end{proof}

Thus the modified objective function is given by
\[
    \psi(x) \coloneqq \sum_{i \in A} \eta\left(u_i(x) - c_i; \frac{1}{2 n^2 (1 + 1 / \delta) \kappa}\right) + \sum_{j \in J} \eta\left(w_j(x) - d_j; \frac{1}{2 n^2 (1 + 1 / \delta) \kappa}\right)
\]
and its gradient is easily computable.
As such we may apply conditional gradient as before to obtain the following theorem.

\begin{theorem}
    Algorithm~\ref{alg:fw_eg_splc_mod} returns some $x$ with $\psi(x^*) - \psi(x) \leq \epsilon$ in $O\left(\frac{n^5 (1 + 1 / \delta)^2 \kappa^2}{\epsilon}\right)$ many iterations.
    Each iteration can be implemented in $O(n^4 m^2 \log n)$ time.
\end{theorem}

\begin{proof}
    The proof is identical to that of Theorem~\ref{thm:splc_main} with the only difference being the fact that $\psi$ is $8 n^4 (1 + 1 / \delta)^2 \kappa^2$-smooth which is where the change in the number of iterations comes from.
    Note that the polytope over which we are optimizing has not changed so each iteration can be implemented in the same way as before.
\end{proof}